\documentclass[journal]{./TeX/IEEEtran}
\ifCLASSINFOpdf
  \usepackage[pdftex]{graphicx}
\else
\fi

\usepackage[colorlinks]{hyperref}
\usepackage{multirow}
\usepackage{times}
\usepackage{amsmath,amssymb}
\usepackage{algorithm}
\usepackage{epstopdf}
\usepackage{subfigure}
\usepackage[dvipsnames]{xcolor}
\usepackage{algorithm}
\usepackage{algorithmic}
\usepackage{cite}
\usepackage{fancyvrb}

\usepackage{bbm}
\newtheorem{theorem}{Theorem}

\newtheorem{lemma}{Lemma}
\newtheorem{corollary}{Corollary}

\newtheorem{definition}{Definition}
\newtheorem{example}{Example}

\newenvironment{proof}{\paragraph*{Proof}}{\hfill$\square$}

{\begin{list}               %    with flush left bullets.
    {$\bullet$ \hfill}{
        \setlength{\leftmargin}{\parindent}
        \setlength{\parsep}{0.04\baselineskip}
        \setlength{\itemsep}{0.5\parsep}
        \setlength{\labelwidth}{\leftmargin}
        \setlength{\labelsep}{0em}}
    }
{\end{list}}

\providecommand{\eref}[1]{\eqref{#1}}  % call \eqref from amstex
\providecommand{\cref}[1]{Chapter~\ref{#1}}

\providecommand{\fref}[1]{Figure~\ref{#1}}

\providecommand{\R}{\ensuremath{\mathbb{R}}}
\providecommand{\I}{\ensuremath{\mathbb{I}}}

\providecommand{\E}{\ensuremath{\mathbb{E}}}
\providecommand{\N}{\ensuremath{\mathbb{N}}}

\providecommand{\bydef}{\overset{\text{def}}{=}}

\renewcommand{\vec}[1]{\ensuremath{\boldsymbol{#1}}}
\providecommand{\mat}[1]{\ensuremath{\boldsymbol{#1}}}

\providecommand{\mY}{\mat{Y}}

\providecommand{\vy}{\vec{y}}

\providecommand{\thetahat}{\widehat{\theta}}

\providecommand{\Var}{\mathrm{Var}}

\newcommand{\argmax}[1]{\mathop{\underset{#1}{\mbox{argmax}}}}

\begin{document}

%
% paper title
% Titles are generally capitalized except for words such as a, an, and, as,
% at, but, by, for, in, nor, of, on, or, the, to and up, which are usually
% not capitalized unless they are the first or last word of the title.
% Linebreaks \\ can be used within to get better formatting as desired.
% Do not put math or special symbols in the title.
\title{Exposure-Referred Signal-to-Noise Ratio \\ for Digital Image Sensors}
%
%
% author names and IEEE memberships
% note positions of commas and nonbreaking spaces ( ~ ) LaTeX will not break
% a structure at a ~ so this keeps an author's name from being broken across
% two lines.
% use \thanks{} to gain access to the first footnote area
% a separate \thanks must be used for each paragraph as LaTeX2e's \thanks
% was not built to handle multiple paragraphs
%

\author{Abhiram~Gnanasambandam,~\IEEEmembership{Student~Member,~IEEE}, and~Stanley~H.~Chan,~\IEEEmembership{Senior~Member,~IEEE}% <-this % stops a space
\thanks{A. Gnanasambandam and S. Chan are with the School of Electrical and Computer
Engineering, Purdue University, West Lafayette, IN 47907, USA. Email: \{agnanasa, stanchan\}@purdue.edu.}% <-this % stops a space
%\thanks{}% <-this % stops a space
%\thanks{Manuscript received April 19, 2005; revised August 26, 2015.}
}

\maketitle

% As a general rule, do not put math, special symbols or citations
% in the abstract or keywords.
\begin{abstract}
The signal-to-noise ratio (SNR) is a fundamental tool to measure the performance of an image sensor. However, confusions sometimes arise between the two types of SNRs. The first one is the output-referred SNR which measures the ratio between the signal and the noise seen at the sensor's output. This SNR is easy to compute, and it is linear in the log-log scale for most image sensors. The second SNR is the exposure-referred SNR, also known as the input-referred SNR. This SNR considers the noise at the input by including a derivative term to the output-referred SNR. The two SNRs have similar behaviors for sensors with a large full-well capacity. However, for sensors with a small full-well capacity, the exposure-referred SNR can capture some behaviors that the output-referred SNR cannot.

While the exposure-referred SNR has been known and used by the industry for a long time, a theoretically rigorous derivation from a signal processing perspective is lacking. In particular, while various equations can be found in different sources of the literature, there is currently no paper that attempts to assemble, derive, and organize these equations in one place. This paper aims to fill the gap by answering four questions: (1) How is the exposure-referred SNR derived from first principles? (2) Is the output-referred SNR a special case of the exposure-referred SNR, or are they completely different? (3) How to compute the SNR efficiently? (4) What utilities can the SNR bring to solving imaging tasks? New theoretical results are derived for image sensors of \emph{any} bit-depth and full-well capacity.
\end{abstract}

% Note that keywords are not normally used for peerreview papers.

\section{Introduction}
The signal-to-noise ratio (SNR) is a basic tool to measure a device's performance when acquiring, transmitting, and processing raw data in the presence of noise. In as early as 1949, when Claude Shannon derived the information capacity of a noisy Gaussian channel, the concept of SNR was already presented \cite{Shannon_1949}. As the name suggests, the SNR is the ratio between the signal power and the noise power
\begin{equation}
\text{SNR} = \frac{\text{signal power}}{\text{noise power}},
\end{equation}
which is sometimes expressed in the logarithmic scale via $10\log_{10} \text{SNR}$ with the unit decibel (dB). Assuming that the signal follows the equation $Y = \theta + W$ where $\theta$ is a scalar and $W \sim \text{Gaussian}(0,\sigma^2)$ is the white noise, one can measure the SNR at the \emph{output} by defining
\begin{equation}
\text{SNR}_{\text{out}}(\theta) = \frac{\text{signal at output}}{\text{noise at output}} = \frac{\E[Y]}{\sqrt{\Var[Y]}},
\label{eq: SNR out}
\end{equation}
where the conversion from power to magnitude is taken care by changing the log from $10\log_{10} \text{SNR}$ to $20\log_{10} \text{SNR}_{\text{out}}$. In this equation, $\E[\cdot]$ denotes the expectation and $\Var[\cdot]$ denotes the variance of the random variable $Y$. The SNR defined in \eref{eq: SNR out} is known as the \emph{output-referred} SNR.

In the sensor's literature, there is an alternative definition of the SNR which is based on the exposure (i.e., the input). The definition of this \emph{exposure-referred} SNR is \cite{EMVA_2010}
\begin{equation}
\text{SNR}_{\text{exp}}(\theta) = \frac{\text{signal at input}}{\text{noise at input}} = \frac{\theta}{\sqrt{\Var[Y]}\cdot \frac{d\theta}{d \E[Y]}}.
\label{eq: SNR exp}
\end{equation}
The denominator of the equation is called the exposure-referred noise or the input-referred noise \cite{Gamal_Lecture}. In computer vision, an early work of Mitsunaga and Nayar in 1999 \cite{Mitsunaga_Nayar} mentioned the equation (Equation \eref{eq: SNR exp}) although they did not give it a name.

The interpretation of the derivative $d\theta/d\E[Y]$ is known to the sensor's community \cite{EMVA_2010}. The argument is that the input-referred noise is computed by back-propagating the signal from the output to the input via the transfer function $d\theta/d\E[Y]$ which approximates the relationship between the input $\theta$ and the output $\E[Y]$ \cite[supplementary report]{elgendy2018optimal}. As mentioned in multiple papers \cite{fossum2013modeling, elgendy2018optimal}, the presence of this derivative term is particularly important for sensors with a small full-well capacity because when a sensor saturates, the SNR has to drop. The exposure-referred SNR can capture this phenomenon whereas the output-referred SNR cannot.

Both versions of the SNRs are widely used in the industry \cite{EMVA_2010}. People use the SNRs to evaluate sensors and to design sensor parameters so that the performance is maximized. However, from a theoretical point of view, besides the usual interpretation of treating the derivative as a transfer function, mathematically there is no rigorous proof trying to unify the two SNRs. The paper by Yang et al. \cite{yang2011bits} introduced the SNR from a statistical estimation perspective, whereas the paper by Elgendy and Chan \cite{elgendy2018optimal} proved the equivalence between the SNR in \cite{yang2011bits} and $\text{SNR}_{\text{exp}}$ for the case of one-bit sensors under a zero read noise. The connections between $\text{SNR}_{\text{exp}}$ and $\text{SNR}_{\text{out}}$ and their generalizations to multi-bit sensors remain unanswered. The goal of this paper is to fill the gap by answering four questions:
\begin{enumerate}
\item[(i)] How is the exposure-referred SNR derived from the first principle?
\item[(ii)] Under the unified framework addressed in (i), is it possible to show that $\text{SNR}_{\text{out}}$ is a special case of $\text{SNR}_{\text{exp}}$?
\item[(iii)] How to efficiently compute $\text{SNR}_{\text{exp}}$ for complex forward imaging models?
\item[(iv)] What are the utilities of $\text{SNR}_{\text{exp}}(\theta)$?
\end{enumerate}

\section{Mathematical Background}

\subsection{Limitations of Output-Referred SNR}
To begin the discussion, it would be useful to first study the output-referred SNR and understand its limitations. Consider a Poisson-Gaussian random variable which is commonly used in modeling image sensors \cite{Nakamura_2005_book, lim2006characterization}:
\begin{equation*}
X \sim \text{Poisson}(\theta) + \text{Gaussian}(0,\sigma_{\text{read}}^2),
\end{equation*}
where $\theta$ denotes the average number of photons (i.e., the flux integrated over the surface area and the exposure time), and $\sigma_{\text{\scriptsize read}}$ is the standard deviation of the read noise. This equation is a simplified model. Factors such as dark current, pixel non-uniformity, defects, analog-to-digital conversion, and color filter arrays are not considered to make the theoretical derivations tractable. Later in the paper, when the Monte Carlo simulation is introduced, some of these factors will be included.

Assuming that the sensor has a full-well capacity $L$ (with the unit electrons), the output produced by the sensor follows the equation
\begin{equation}
Y =
\begin{cases}
X, &\qquad X < L,\\
L, &\qquad X \ge L.
\end{cases}
\label{eq: truncated Poisson}
\end{equation}
If the full-well capacity is $L =\infty$, the measurement $Y$ will never saturate, and hence the expectation is $\E[Y] = \theta$ and the variance is $\Var[Y] = \theta + \sigma_{\text{\scriptsize read}}^2$. Then, according to \eref{eq: SNR out}, $\text{SNR}_{\text{out}}$ can be computed via
\begin{equation}
\text{SNR}_{\text{out}}(\theta) = \frac{\theta}{\sqrt{\theta+\sigma_{\text{read}}^2}}.
\label{eq: SNR out 2}
\end{equation}
If the read noise is negligible, i.e. $\sigma_{\text{\scriptsize read}} \ll \theta$, the SNR can be simplified to $\text{SNR}_{\text{out}}(\theta) = \sqrt{\theta}$. This equation can be found in many standard texts, e.g., \cite{Nakamura_2005_book}.

The problem arises when the full-well capacity is \emph{finite}. When $L < \infty$, $\text{SNR}_{\text{out}}$ in \eref{eq: SNR out} is still adequate to capture the sensor's behavior when the exposure $\theta$ is smaller than the full-well capacity $L$. However, if $\theta$ reaches the full-well and goes beyond it, the mean $\E[Y]$ will stop growing with $\theta$ as illustrated in \fref{fig: response}. The variance $\Var[Y]$ will gradually drop to zero because when $Y$ goes beyond the full-well capacity, it will be capped. As a result, $\text{SNR}_{\text{out}}$ according to \eref{eq: SNR out} will eventually go to infinity (because $\Var[Y] \rightarrow 0$). This is not a desirable behavior because the SNR beyond the saturation is supposed to be poor.

\begin{figure}[ht]
\centering
\includegraphics[width=\linewidth]{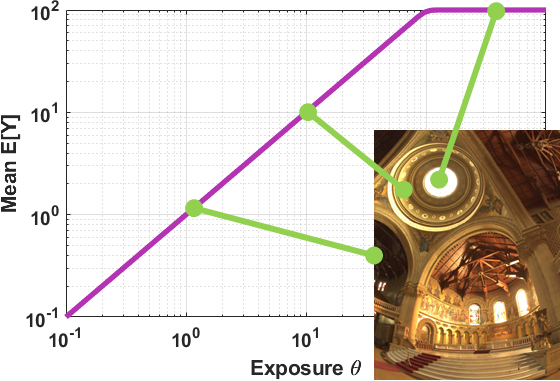}
\vspace{-4ex}
\caption{With a finite full-well capacity, the mean $\E[Y]$ will stop growing when the exposure $\theta$ exceeds the full-well capacity $L = 10^2$.}
\label{fig: response}
\end{figure}

\subsection{One-bit and Multi-bit Sensors}
If the full-well capacity is the source of the issue, the next question is how commonly does it happen. The full-well capacity of a conventional CMOS image sensor (CIS) is thousands of electrons. Therefore, unless the scene is bright and the integration is long, the sensor generally does not need to operate near the full-well capacity. However, for a newer type of image sensors based on single-photon detectors, the full-well capacity can be as small as only one electron. Their idea is to acquire binary bits at a very high frame rate and use computational methods to reconstruct the image. This class of sensors is broadly known as the quanta image sensors (QIS) that can be implemented using the single-photon avalanche diodes (SPAD-QIS) or the existing CMOS technology (CIS-QIS) \cite{fossum200611}.

From a mathematical modeling point of view, QIS is no different from a conventional CIS if it operates in the multi-bit mode because the equation will follow \eref{eq: truncated Poisson} \cite{fossum2013modeling}. The difference is that for QIS, the read noise is about 0.2 electrons whereas for CIS, the read noise can range from 2 to tens of electrons. If QIS operates in the one-bit mode, $Y$ will follow
\begin{equation}
Y =
\begin{cases}
0, &\qquad X < q,\\
1, &\qquad X \ge q,
\end{cases}
\label{eq: QIS model}
\end{equation}
where $q$ is the threshold (usually set as $q = 0.5$). Because of the generality of equations \eref{eq: truncated Poisson} and \eref{eq: QIS model}, the theoretical results of this paper is applicable to sensors of \emph{any} bit-depth, including QIS and the conventional CIS.

\subsection{Truncated Poisson and the Incomplete Gamma Function}
The multi-bit sensor equation in \eref{eq: truncated Poisson} requires some statistical tools to handle the truncated Poisson random variable. To further simplify notations, in this section, the read noise is $\sigma_{\text{read}} = 0$ so that $X \sim \text{Poisson}(\theta)$. In this case, the probability mass function of $Y$ is
\begin{equation}
p_Y(y) =
\begin{cases}
\frac{\theta^{y}}{y!}e^{-\theta}, &\qquad y < L,\\
\sum_{\ell = L}^\infty \frac{\theta^{\ell}}{\ell!}e^{-\theta}, &\qquad y = L.
\end{cases}
\end{equation}
By construction, the random variable $Y$ will never take a value greater than $L$. The probability that $Y = L$ is given by the sum of the Poisson tail, which can be conveniently expressed via the incomplete Gamma function as shown in \fref{fig: incomplete gamma function}.
\begin{figure}[ht]
\centering
\vspace{-2ex}
\includegraphics[width=0.95\linewidth]{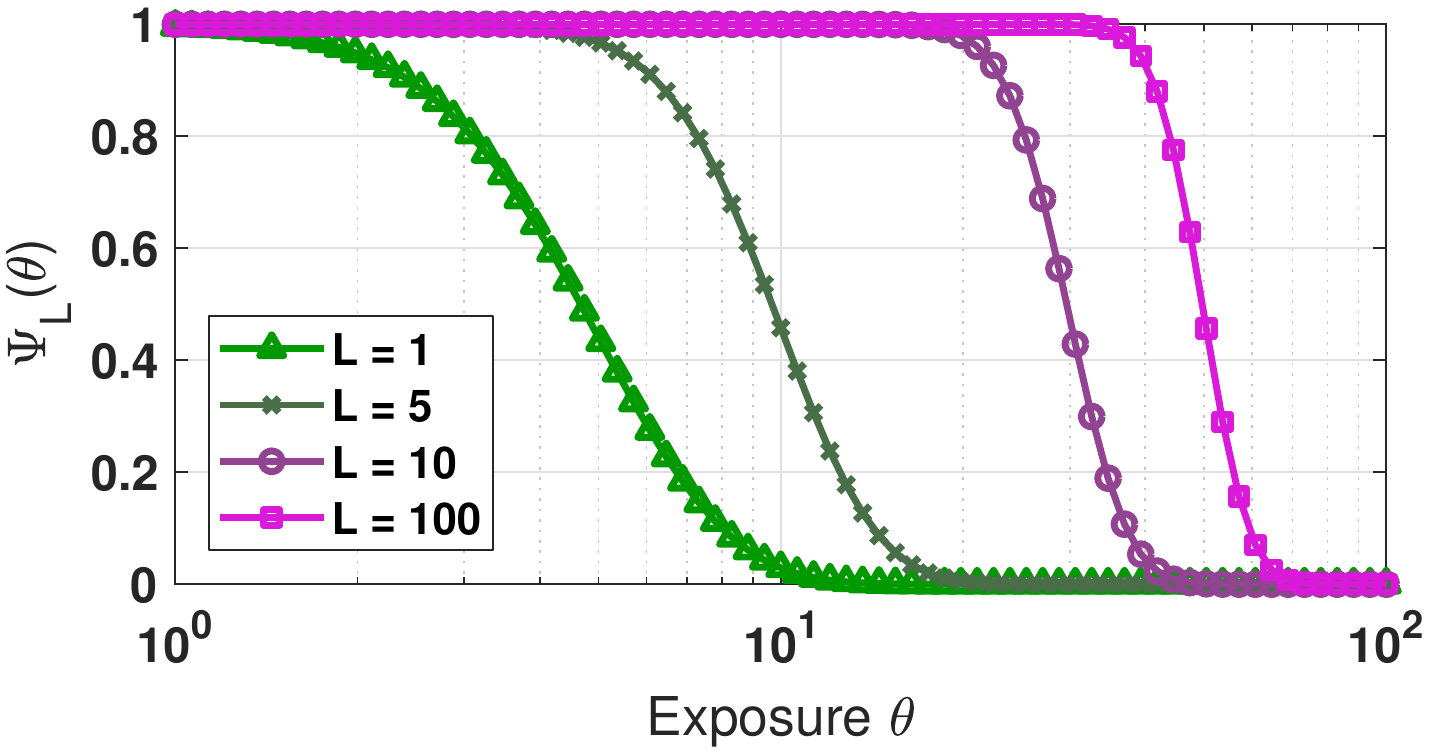}
\vspace{-2ex}
\caption{Incomplete Gamma function $\Psi_L(\theta)$ as a function of $\theta$.}
\vspace{-2ex}
\label{fig: incomplete gamma function}
\end{figure}

\begin{definition}[Incomplete Gamma function, \cite{Whittlesey}]
The upper incomplete Gamma function is defined as $\Psi_L: \R_+ \rightarrow [0,1]$, with
\begin{equation}
\Psi_L(\theta) = \frac{1}{\Gamma(L)}\int_{\theta}^{\infty}t^{L-1}e^{-t}dt = \sum_{\ell=0}^{L-1}\frac{\theta^\ell e^{-\theta}}{\ell!},
\end{equation}
for $\theta>0, L \in \N$ where $\Gamma(L) = (L-1)!$ is the standard Gamma function.
\end{definition}

A few useful properties of $\Psi_L(\theta)$ can be derived. First, the first-order derivative of $\Psi_L(\theta)$ is
\begin{equation}
\Psi'_L(\theta) = -\frac{\theta^{L-1}e^{-\theta}}{(L-1)!} < 0, \;\; \text{for all} \; \theta,
\end{equation}
which means that $\Psi_L(\theta)$ is a strictly decreasing function in $\theta$. The steepest slope can be determined by analyzing the curvature
\begin{equation*}
\Psi''_L(\theta) = -(L-1)\theta^{L-2}e^{-\theta} + e^{-\theta}\theta^{L-1}.
\end{equation*}
Equating this to zero will yield $\theta^* = L-1$. At this critical point and assume $L \gg 1$, a new result using Stirling's formula can be shown\footnote{The proof of this result is given in the Appendix.}:
\begin{align}
\Psi_L'(\theta^*) \approx -\frac{1}{\sqrt{2\pi\theta^*}} \exp\left\{-\frac{(\theta^*-(L-1))^2}{2\theta^*}\right\}.
\label{eq: Psi approximate}
\end{align}
Therefore, $\Psi_L'(\theta^*) = -\frac{1}{\sqrt{2\pi(L-1)}}$. Hence, the slope of the incomplete Gamma function reduces as $L$ increases.

\emph{Remark 1}: Most papers in the image sensor's literature plot curves with respect to $\log_{10} \theta$ instead of $\theta$, like the one shown in \fref{fig: incomplete gamma function}. The $x$-axis compression caused by $\log_{10}\theta$ will make the transient of the incomplete Gamma function to appear steeper. The reason is that for any function $f(\theta)$, the slope in the $\log_{10} \theta$ space is determined by $\frac{d}{d\log_{10} \theta} f(\theta) = \theta f'(\theta) \cdot \log 10$. So for large $\theta$, the slope appears steeper.

\subsection{Delta Method}
A mathematical tool that will become useful later in the paper is the \emph{Delta Method} in statistics. It approximates the variance when a random variable undergoes a nonlinear transformation.

\begin{lemma}[Delta Method, \cite{Lehmann_1999}]
Consider a sequence of independently and identically distributed (i.i.d.) random variables $X_1,\ldots,X_N$ with a common mean $\E[X_1] = \mu$. Define $\overline{X} = (1/N)\sum_{n=1}^N X_n$ be the sample average, and assume that $\sqrt{N}(\overline{X}-\mu) \overset{d}{\rightarrow} \text{Gaussian}(0,\tau^2)$, where $\tau^2$ is the standard deviation of a Gaussian random variable, and $\overset{d}{\rightarrow}$ denotes convergence in distribution. Suppose that there is a continuously differentiable function $f$ such that $f'(\mu)$ exists and is not zero. Then $\sqrt{N}[f(\overline{X}) - f(\mu)] \overset{d}{\rightarrow} \text{Gaussian}(0,\tau^2 (f'(\mu))^2)$. In other words,
\begin{equation}
\E[(f(\overline{X})-f(\mu))^2] \approx [f'(\mu)]^2 \Var[\overline{X}].
\label{eq: delta method equation}
\end{equation}
\end{lemma}
\begin{proof}
The complete proof can be found in \cite[Theorem 2.5.2]{Lehmann_1999}. The two key arguments in the proof are (1) $\overline{X}$ converges in distribution to $\mu$ due to the Central Limit Theorem; (2) Taylor approximation gives
\begin{align*}
f(\overline{X}) \approx f(\mu) + f'(\mu)(\overline{X}-\mu) + o_p(\overline{X}-\mu),
\end{align*}
where $o_p$ denotes the little-$o_p$ notation \cite[Definition 2.1.3]{Lehmann_1999} and hence
\begin{align*}
\sqrt{N}\left[f(\overline{X}) - f(\mu)\right] \approx \sqrt{N} f'(\mu)(\overline{X} - \mu) + o_p(\sqrt{N}(\overline{X}-\mu)).
\end{align*}
Taking squares on both sides gives \eref{eq: delta method equation}.
\end{proof}

\section{SNR: A Statistical Definition}
\subsection{Defining the SNR}
When defining the SNR, it is important to clarify the notion of signal and noise. In most of the imaging problems, the underlying signal is the scene exposure $\theta$. The signal $\theta$ defines a probability distribution $p_Y(y;\theta)$ from which the i.i.d. samples $Y_1,\ldots,Y_N$ are drawn.

Reconstruction of the signal $\theta$ from $Y_1,\ldots,Y_N$ is based on an \emph{estimator} $\thetahat$. An estimator can be any mapping that maps $\mY = [Y_1,\ldots,Y_N]$ to $\thetahat(\mY)$. However, if $\theta$ is a fixed scalar, most sensors will produce an estimate based on the sample average $\overline{Y} = (1/N)\sum_{n=1}^N Y_n$ because the on-chip processing today is largely limited to simple operations such as addition. In this case, the estimator can be written as $\widehat{\theta}(\overline{Y})$, which can also be interpreted as an estimator based on the sufficient statistics.

The noise term in the SNR is the deviation between the estimator $\thetahat(\mY)$ and the true parameter $\theta$. Since the estimator $\thetahat(\mY)$ is random, the noise power is the expectation
\begin{equation}
\text{noise power} = \E[(\thetahat(\mY) - \theta)^2],
\end{equation}
which is also the mean squared error.

\begin{definition}[SNR, formal definition]
Let $\mY = [Y_1,\ldots,Y_N]$ be i.i.d. random variables drawn from the distribution $p_Y(y;\theta)$. Construct an estimator $\thetahat(\mY)$. Then the signal-to-noise ratio (SNR) is
\begin{equation}
\text{SNR}(\theta) \bydef \frac{\theta}{\sqrt{\E[(\thetahat(\mY) - \theta)^2]}}.
\label{eq: SNR definition}
\end{equation}
\end{definition}

The above definition is not a new invention. When Yang et al. presented the analysis of the quanta image sensor in 2011 \cite{yang2011bits}, this definition was already used to derive the Cramer-Rao lower bound. Subsequent papers such as \cite{elgendy2018optimal} and \cite{Gnanasambandam_TCI_HDR} by Chan and colleagues also rely on this definition. A version for single-photon avalanche diode (SPAD) is presented by Gupta and colleagues \cite{Gupta_2019}.

To elaborate on this formal definition of the SNR, consider the two examples below.
\begin{example}[Poisson]
Let $Y_1,\ldots,Y_N \overset{\text{\tiny i.i.d.}}{\sim} \text{Poisson}(\theta)$, and consider the maximum-likelihood estimator $\thetahat(\overline{Y}) = \overline{Y}$. (The derivation of the ML estimator for a Poisson random variable is skipped for brevity.) Since $\thetahat(\overline{Y}) = \overline{Y}$, it follows that
\begin{equation*}
\E[(\thetahat(\overline{Y}) - \theta)^2] = \E[(\overline{Y}-\theta)^2] = \frac{\theta}{N},
\end{equation*}
where the second equality holds because the variance of a Poisson random variable is $\theta$. Therefore, $\text{SNR}(\theta) = \sqrt{N} \cdot \sqrt{\theta}$, which is consistent with \eref{eq: SNR out 2} when $\sigma_{\text{\scriptsize read}} = 0$. \hfill $\square$
\end{example}

\begin{example}[Poisson + Gaussian]
Let $Y_1,\ldots,Y_N \overset{\text{\tiny i.i.d.}}{\sim} \text{Poisson}(\theta) + \text{Gaussian}(0,\sigma_{\text{\scriptsize read}}^2)$. Consider the ML estimator $\thetahat(\overline{Y}) = \overline{Y}$. It then follows that
\begin{equation*}
\E[(\thetahat(\overline{Y}) - \theta)^2] = \E[(\overline{Y}-\theta)^2] = \frac{1}{N} \left(\theta+\sigma_{\text{\scriptsize read}}^2\right),
\end{equation*}
where the second equality holds because the variance of a Poisson-Gaussian is the sum of the two variances. Therefore, $\text{SNR}(\theta) = \sqrt{N} \cdot \theta / \sqrt{\theta+\sigma_{\text{\scriptsize read}}^2}$. This result is consistent with \eref{eq: SNR out 2} for a general $\sigma_{\text{\scriptsize read}}$.  \hfill $\square$
\end{example}

\subsection{Mean Invariance}
The formal definition of the SNR is general for \emph{any} estimator. However, the statistical noise model for an actual image sensor can be complicated. In fact, except for a few special occasions where the maximum-likelihood (ML) estimator can be expressed in a closed form, in most other situations the ML estimator cannot be obtained in a closed form. In this case, a more convenient way is to define the estimator from the mean.

\begin{definition}[Mean invariant estimator]
Let $Y_1,\ldots,Y_N$ be i.i.d. random variables drawn from the distribution $p_Y(y;\theta)$. Let $\mu(\theta) = \E[Y] = \int y \cdot p_Y(y; \theta) \, dy$ be the mean of $Y$. An estimator $\thetahat(\mY)$ is mean invariant if
\begin{equation}
\mu(\thetahat(\mY)) = \overline{Y}.
\end{equation}
If $\mu^{-1}$ exists, the mean invariant estimator is $\thetahat(\mY) = \mu^{-1}(\overline{Y})$.
\end{definition}

The following two examples shows that many estimators are mean invariant.
\begin{example}
\label{example MI Gaussian}
Let $Y_n \overset{ \text{\tiny{i.i.d.}}}{\sim} \text{Gaussian}(\theta,\sigma_{\text{\scriptsize read}}^2)$ for $n = 1,\ldots,N$, where $\theta$ is the unknown parameter. It can be shown that the maximum-likelihood (ML) estimator is
\begin{align*}
\thetahat_{\text{ML}}(\overline{Y}) = \argmax{\theta}\; \prod_{n=1}^N \frac{1}{\sqrt{2\pi\sigma_{\text{\scriptsize read}}^2}} \exp\left\{-\frac{(Y_n-\theta)^2}{2\sigma_{\text{\scriptsize read}}^2}\right\} = \overline{Y}.
\end{align*}
Notice that if $Y$ is Gaussian, the mean is $\mu(\theta) = \E[Y] = \theta$. Therefore, the mean invariance property holds:
\begin{align*}
\mu(\thetahat_{\text{ML}}(\overline{Y})) \overset{(a)}{=} \thetahat_{\text{ML}}(\overline{Y}) = \overline{Y}.
\end{align*}
where $(a)$ is due to the fact that $\mu(\theta) = \theta$.  \hfill $\square$
\end{example}

\begin{example}
\label{example MI one bit}
Let $Y_1,\ldots,Y_N$ be i.i.d. one-bit measurements defined according to \eref{eq: QIS model}, with $q = 0.5$ and $\sigma_\text{read} = 0$. In this case, $Y_n$ is a Bernoulli random variable such that $Y_n \overset{ \text{\tiny{i.i.d.}}}{\sim} \text{Bernoulli}(1-e^{-\theta})$. The ML estimator is
\begin{align*}
\thetahat_{\text{ML}}(\overline{Y})
&= \argmax{\theta}\; \prod_{n=1}^N (1-e^{-\theta})^{Y_n} (e^{-\theta})^{1-Y_n} \\
&= -\log(1-\overline{Y}).
\end{align*}
Since the mean of $Y$ is $\mu(\theta) = \E[Y] = 1-e^{-\theta}$, it follows that
\begin{equation*}
\mu(\thetahat_{\text{ML}}(\overline{Y})) = 1-e^{-\thetahat_{\text{ML}}(\overline{Y})} = \overline{Y}.
\end{equation*}
Again, the mean invariance property is satisfied. \hfill $\square$
\end{example}

A mean invariant estimator is easy to construct. Even if $p_Y(y;\theta)$ has a complex form, the mean $\E[Y]$ can be obtained through Monte Carlo simulation. Once the mean $\E[Y]$ is determined, the mean invariant estimator $\thetahat$ can be \emph{constructed} from $\thetahat(\overline{Y}) = \mu^{-1}(\overline{Y})$, assuming that $\mu^{-1}$ exists.

\emph{Follow up of Example~\ref{example MI Gaussian}}. Let $Y_n \overset{ \text{\tiny{i.i.d.}}}{\sim} \text{Gaussian}(\theta,\sigma_{\text{read}}^2)$ for $n = 1,\ldots,N$ where $\theta$ is the unknown parameter. Since the mean is $\E[Y] = \theta$, it follows that $\mu(\theta) = \E[Y] = \theta$. This $\mu$ is the identity mapping, and so the inverse mapping is $\mu^{-1}(s) = s$ for any $s$. Thus, one can define an estimator $\thetahat(\overline{Y}) = \mu^{-1}(\overline{Y}) = \overline{Y}$. As seen, it is the same as the ML estimator derived in Example~\ref{example MI Gaussian}. Moreover, $\thetahat(\overline{Y})$ is mean invariant because $\thetahat$ is constructed in that way.  \hfill $\square$

\emph{Follow up of Example~\ref{example MI one bit}}. Let $Y_n \overset{ \text{\tiny{i.i.d.}}}{\sim} \text{Bernoulli}(1-e^{-\theta})$ be the one-bit measurements for $n = 1,\ldots,N$. The mean is $\mu(\theta) = \E[Y] = 1-e^{-\theta}$, and so the inverse is $\mu^{-1}(s) = -\log(1-s)$. Therefore, one can define the estimator as $\thetahat(\overline{Y}) = \mu^{-1}(\overline{Y}) = -\log(1-\overline{Y})$ The result is identical to Example~\ref{example MI one bit}. Furthermore, since the estimator is constructed from the mean invariance property, it has to satisfy the property. \hfill $\square$

Based on Example~\ref{example MI Gaussian} and Example~\ref{example MI one bit}, one may conjecture that any ML estimator is also the mean invariant estimator. The observation is correct for any distributions in the exponential family. The proof is given in the Appendix. Outside the exponential family the two can be different. The following is a counter example.

\begin{example}[Mean invariant estimator $\not=$ ML estimator]
\label{example: MIE MLE}
Consider a truncated Poisson distribution
\begin{align*}
p_Y(y;\theta) =
\begin{cases}
\frac{\theta^y e^{-\theta}}{y!}, &\qquad y < L,\\
1-\Psi_L(\theta), &\qquad y \ge L.
\end{cases}
\end{align*}

Using \eref{eq: SNR exp main components} (to be proved in the next section), the mean is $\E[Y] = \theta \Psi_{L-1}(\theta) + L (1-\Psi_L(\theta))$. Let $\mu(\theta) = \theta \Psi_{L-1}(\theta) + L (1-\Psi_L(\theta))$. So, the mean invariant estimator can be defined as
\begin{equation*}
\thetahat(\mY) = \mu^{-1}(\overline{Y}),
\end{equation*}
which is a function of $\overline{Y} = (1/N)\sum_{n=1}^N Y_n$.

Now consider the ML estimator. The ML estimator is
\begin{align}
\thetahat_{\text{ML}}(\mY)
&= \argmax{\theta} \;\; \frac{1}{N} \sum_{n=1}^N  \Bigg[( Y_n \log \theta - \theta ) \cdot \underset{Z_n}{\underbrace{\I\{Y_n < L\}}} \notag\\
& \quad + \log (1-\Psi_L(\theta)) \cdot  \underset{1-Z_n}{\underbrace{\I\{Y_n \ge L\}}} \Bigg],
\label{eq: ML MIE condition 1}
\end{align}
where $Z_n = \I\{Y_n < L\}$ is the indicator function that returns 1 if $Y_n < L$ or 0 if otherwise. Taking derivative and setting it to zero implies that $\thetahat_{\text{ML}}$ must satisfy the equation
\begin{align}
&\frac{1}{\thetahat_{\text{ML}}} \left(\frac{1}{N} \sum_{n=1}^N Y_n Z_n \right) - \frac{1}{N} \sum_{n=1}^N Z_n \notag \\
&\qquad = \left(1 - \frac{1}{N} \sum_{n=1}^N Z_n\right) \cdot \frac{\Psi_L'(\thetahat_{\text{ML}})}{1-\Psi_L(\thetahat_{\text{ML}})}.
\label{eq: ML MIE condition}
\end{align}
Therefore, the ML estimator $\thetahat_{\text{ML}}(\mY)$ must be a function of $(1/N)\sum_{n=1}^N Y_n Z_n$ and $(1/N)\sum_{n=1}^N Z_n$, not $\overline{Y}$. \hfill $\square$
\end{example}

While the ML estimator and the mean invariant estimator are generally not the same, they are asymptotically equivalent. As $N \rightarrow \infty$, the consistency of the ML estimator implies that $\thetahat_{\text{ML}}(\mY) \overset{p}{\rightarrow} \theta$ \cite{Lehmann_1999}. On the other hand, the law of large numbers implies that $\overline{Y} \overset{p}{\rightarrow} \E[Y] = \mu(\theta)$. So, a mean invariant estimator $\thetahat(\mY) = \mu^{-1}(\overline{Y}) \overset{p}{\rightarrow} \mu^{-1}(\mu(\theta)) = \theta$. Therefore, as $N \rightarrow \infty$, both $\thetahat_{\text{ML}}(\mY)$ and $\thetahat(\mY)$ will converge in probability to the true parameter $\theta$ and hence they are equivalent asymptotically.

Mean invariance is also not the same as the invariance principle of the ML estimator. The invariance principle of the ML estimator says that if there is a monotonic mapping $h$ that maps $\theta$ to $h(\theta)$, then $h(\widehat{\theta}_{\text{ML}})$ will be the ML estimator of $h(\theta)$. The following example shows a case where mean invariance is different from the invariance of ML.

\begin{example}
Consider $X_n \overset{ \text{\tiny{i.i.d.}}}{\sim} \text{Bernoulli}(\theta)$ and $Y_n \overset{ \text{\tiny{i.i.d.}}}{\sim} \text{Bernoulli}(e^\theta)$ for $n = 1,\ldots,N$. The ML estimator of $\theta$ using $X_1,\ldots,X_N$ is
\begin{equation*}
\widehat{\theta}_{\text{ML}}(\overline{X}) = \overline{X}.
\end{equation*}
According to the invariance of the ML estimator, a monotonic mapping $h_1(\cdot) = e^{(\cdot)}$ will ensure that $\widehat{\theta}_1 \bydef h_1(\widehat{\theta}_{\text{ML}}) = e^{\overline{X}}$ is the ML estimator of $e^{\theta}$. By the same principle, if there is a different monotonic mapping $\widehat{\theta}_2 \bydef h_2(\cdot) = \log(\cdot)$, it holds that $h_2(\widehat{\theta}_{\text{ML}}) = \log(\overline{X})$ is the ML estimator of $\log(\theta)$.

Now consider the mean $\mu(\theta) \bydef \E[Y] = e^{\theta}$. The invariance principle says nothing about whether $\mu(\widehat{\theta}_{1}(\overline{Y})) = \overline{Y}$ or $\mu(\widehat{\theta}_{2}(\overline{Y})) = \overline{Y}$. In fact, $\widehat{\theta}_{1}$ does not satisfy the mean invariance property because $\mu(\widehat{\theta}_{1}(\overline{Y})) = e^{\thetahat_1(\overline{Y})} = e^{e^{\overline{Y}}} \not= \overline{Y}$. However, $\widehat{\theta}_{2}$ satisfies the mean invariance property because $\mu(\widehat{\theta}_{2}(\overline{Y})) = e^{\log(\overline{Y})} = \overline{Y}$. Therefore, the invariance principle of the ML estimator is completely different from the mean invariance property. \hfill $\square$
\end{example}

To summarize, the mean invariance is a property that specifically focuses on whether the mean $\E[Y]$ can be nonlinearly mapped to recover the true parameter $\theta$. This is the property required for the exposure-referred SNR. Whether the estimator is the ML estimator is not of concern.

In the statistics literature, the mean invariance property presented in this paper is related to the \emph{link function} for the generalized linear models. Specifically, the mapping $\mu(\theta) \bydef \E[Y]$ from the true parameter $\theta$ to the mean $\mu$ is known as the link function, and the inverse mapping $\mu^{-1}$ is known as the response function. Readers interested in details of this topic can consult \cite{McCullagh_Nelder_1983}.

\subsection{Exposure-referred SNR}
With all the mathematical tools ready, the exposure-referred SNR can now be formally derived.

\begin{theorem}
\label{thm: SNR exp}
Let $Y_1,\ldots,Y_N$ be i.i.d. random variables drawn from the probability density function $p_Y(y;\theta)$. Define $\overline{Y} = (1/N)\sum_{n=1}^N Y_n$. Let $\mu(\theta) \bydef \E[Y]$ and assume that $\mu^{-1}$ exists. Let $\thetahat(\mY)$ be the mean invariant estimator such that $\thetahat(\mY) = \mu^{-1}(\overline{Y})$. Then the SNR defined in \eref{eq: SNR definition} is related to $\text{SNR}_{\text{exp}}$ as
\begin{equation}
\text{SNR}(\theta) \approx \text{SNR}_{\text{exp}}(\theta) = \frac{\theta}{\sqrt{\Var[\overline{Y}]}} \cdot \frac{d\mu}{d\theta}.
\end{equation}
\end{theorem}

\begin{proof}
By the Delta Method, the mean squared error can be approximated by
\begin{align*}
\E\left[\left(\thetahat(\overline{Y}) - \theta\right)^2\right]
&= \E\left[\left(\thetahat(\overline{Y}) - \thetahat(\mu(\theta))\right)^2\right] \\
&\approx \left[\thetahat'(\mu(\theta))\right]^2\Var[\overline{Y}],
\end{align*}
where the derivative is taken with respect to $\mu(\theta)$.

Since $\thetahat(\mu(\theta)) = \mu^{-1}(\mu(\theta)) = \theta$, it follows that $\frac{d\thetahat(\mu)}{d\mu} = \frac{d\theta}{d\mu}$. So,
\begin{align}
\E\left[\left(\thetahat(\overline{Y}) - \theta\right)^2\right]  = \left[\frac{d\theta}{d\mu}\right]^2\Var[\overline{Y}].
\label{eq: proof step}
\end{align}
Using the fact that $\frac{d\theta}{d\mu} = 1/\frac{d\mu}{d\theta}$, the SNR can be written as
\begin{align*}
\text{SNR}(\theta) = \frac{\theta}{\sqrt{\E\left[\left(\thetahat(\overline{Y}) - \theta\right)^2\right] }} \approx
\underset{\text{SNR}_{\text{exp}}(\theta)}{\underbrace{\frac{\theta}{\sqrt{\Var[\overline{Y}]}} \cdot \frac{d\mu}{d\theta}}},
\end{align*}
which completes the proof.
\end{proof}

\begin{corollary}
Under the same conditions listed in Theorem~\ref{thm: SNR exp}, the exposure-referred SNR is related to the output-referred SNR as
\begin{equation}
\text{SNR}_{\text{exp}}(\theta) = \text{SNR}_{\text{out}}(\theta) \cdot \frac{\theta}{\mu} \cdot \frac{d\mu}{d\theta}.
\end{equation}
\end{corollary}
\begin{proof}
The proof follows from the substitution
\begin{align*}
\text{SNR}_{\text{exp}}(\theta) = \frac{\theta}{\sqrt{\Var[\overline{Y}]}} \cdot \frac{d\mu}{d\theta}
&= \underset{=\text{SNR}_{\text{out}}(\theta)}{\underbrace{\frac{\mu}{\sqrt{\Var[\overline{Y}]}}}} \cdot \frac{\theta}{\mu} \cdot \frac{d\mu}{d\theta}.
\end{align*}
This completes the proof.
\end{proof}

As one can see from \eref{eq: proof step}, the derivative $d\mu/d\theta$ is added because of the Delta Method. It is the first-order approximation of a nonlinear mapping from the output to the input. Using the argument of Elgendy and Chan \cite{elgendy2018optimal}, this first-order derivative can be regarded as a transfer function relating the output $\E[\overline{Y}]$ to the input $\theta$. If the input-output has a linear relationship $\E[\overline{Y}] = \theta$, which is the case of a CIS with a large full-well capacity, then the derivative is $d\mu/d\theta = 1$ and so $\text{SNR}_{\text{out}}(\theta) = \text{SNR}_{\text{exp}}(\theta)$.

\subsection{Illustrating the SNR via one-bit QIS}
To elaborate on the difference between $\text{SNR}_{\text{exp}}(\theta)$ and $\text{SNR}_{\text{out}}(\theta)$, it would be instructive to consider the statistics of a one-bit quanta image sensor. Let $X_1,\ldots,X_N \overset{\text{\tiny i.i.d.}}{\sim} \text{Poisson}(\theta)$ and let $Y_n$ be the random variables defined in \eref{eq: QIS model}.

First, consider the case where $q = 1$. Since $Y_1,\ldots,Y_N \overset{\text{\tiny i.i.d.}}{\sim} \text{Bernoulli}(1-e^{-\theta})$, the mean is $\E[Y] = 1-e^{-\theta}$. Define the mean as $\mu(\theta) = \E[Y] = 1-e^{-\theta}$. As shown in Example~\ref{example MI one bit}, the maximum-likelihood estimate of $\theta$ is $\thetahat(\overline{Y}) = -\log(1-\overline{Y})$ and it satisfies the mean invariance property. The derivative $d\mu/d\theta$ is
\begin{equation*}
\frac{d\mu}{d\theta} = \frac{d}{d\theta}\left[1-e^{-\theta}\right] = e^{-\theta}.
\end{equation*}
Substituting into Theorem~\ref{thm: SNR exp}, it can be shown that
\begin{align*}
\text{SNR}_{\text{exp}}(\theta)
= \frac{\theta}{\sqrt{\Var[\overline{Y}]}} \cdot \frac{d\mu}{d\theta} = \sqrt{N} \cdot \theta \cdot \sqrt{\frac{e^{-\theta}}{1-e^{-\theta}}}.
\end{align*}
For cases where $q > 1$, one can use the incomplete Gamma function so that
\begin{equation*}
p_Y(y;\theta) =
\begin{cases}
1-\Psi_q(\theta), &\qquad y = 1,\\
\Psi_q(\theta),   &\qquad y = 0.
\end{cases}
\end{equation*}
It then follows that $\E[Y] = 1-\Psi_q(\theta)$ and the estimator can be chosen such that $\thetahat(\overline{Y}) = \Psi_q^{-1}(1-\overline{Y})$. The mean invariance property is therefore validated. The derivative $d\mu/d\theta$ is
\begin{equation*}
\frac{d\mu}{d\theta} = \frac{d}{d\theta}(1-\Psi_q(\theta)) = \frac{\theta^{q-1}e^{-\theta}}{(q-1)!}.
\end{equation*}
Hence, the SNR is
\begin{align}
\text{SNR}_{\text{exp}}(\theta) = \frac{\sqrt{N} \cdot \theta}{\sqrt{\Psi_q(\theta)(1-\Psi_q(\theta))}} \cdot \frac{\theta^{q-1}e^{-\theta}}{(q-1)!},
\label{eq: SNR exp QIS}
\end{align}
of which the visualization is shown in \fref{fig: ch3 SNR exposure}. This result is consistent with the one shown by Elgendy and Chan by deriving the Fisher Information \cite{elgendy2018optimal}.

\begin{figure}[ht]
\centering
\includegraphics[width=\linewidth]{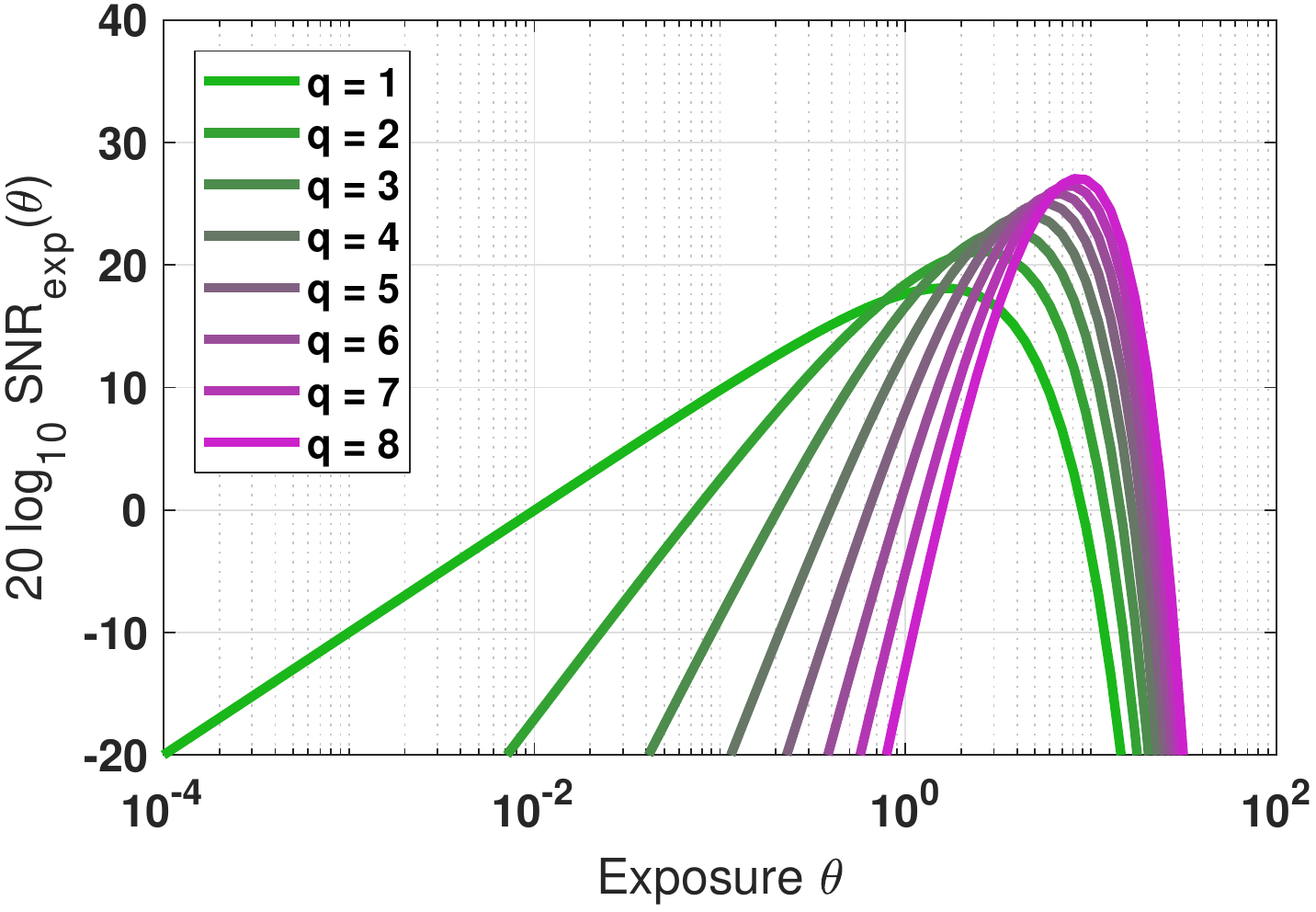}
\vspace{-4ex}
\caption{Exposure-referred SNR for a one-bit quanta image sensor, assuming that $\sigma_{\text{read}} = 0$. As $\theta$ goes beyond the threshold $q$, $\text{SNR}_{\text{exp}}(\theta)$ starts to drop as expected.}
\label{fig: ch3 SNR exposure}
\end{figure}

Unlike $\text{SNR}_{\text{exp}}(\theta)$, the output-referred SNR goes to infinity when $\theta$ grows. For the same one-bit statistics, the output-referred SNR is simply the ratio between $\E[\overline{Y}]$ and $\Var[\overline{Y}]$, which is
\begin{equation}
\text{SNR}_{\text{out}}(\theta) = \frac{\E[\overline{Y}]}{\sqrt{\Var[\overline{Y}]}} = \sqrt{N} \cdot \sqrt{\frac{1-\Psi_q(\theta)}{\Psi_q(\theta)}}.
\end{equation}
As shown in \fref{fig: ch3 SNR output}, $\text{SNR}_{\text{out}}(\theta)$ grows indefinitely as $\theta$ grows. This does not reflect the reality because when $\theta$ grows beyond the threshold $q$, the measurements $\{Y_n \,|\, n = 1,\ldots,N\}$ will have more one's. The signal degrades and hence eventually the SNR drops to zero.

\begin{figure}[ht]
\centering
\includegraphics[width=\linewidth]{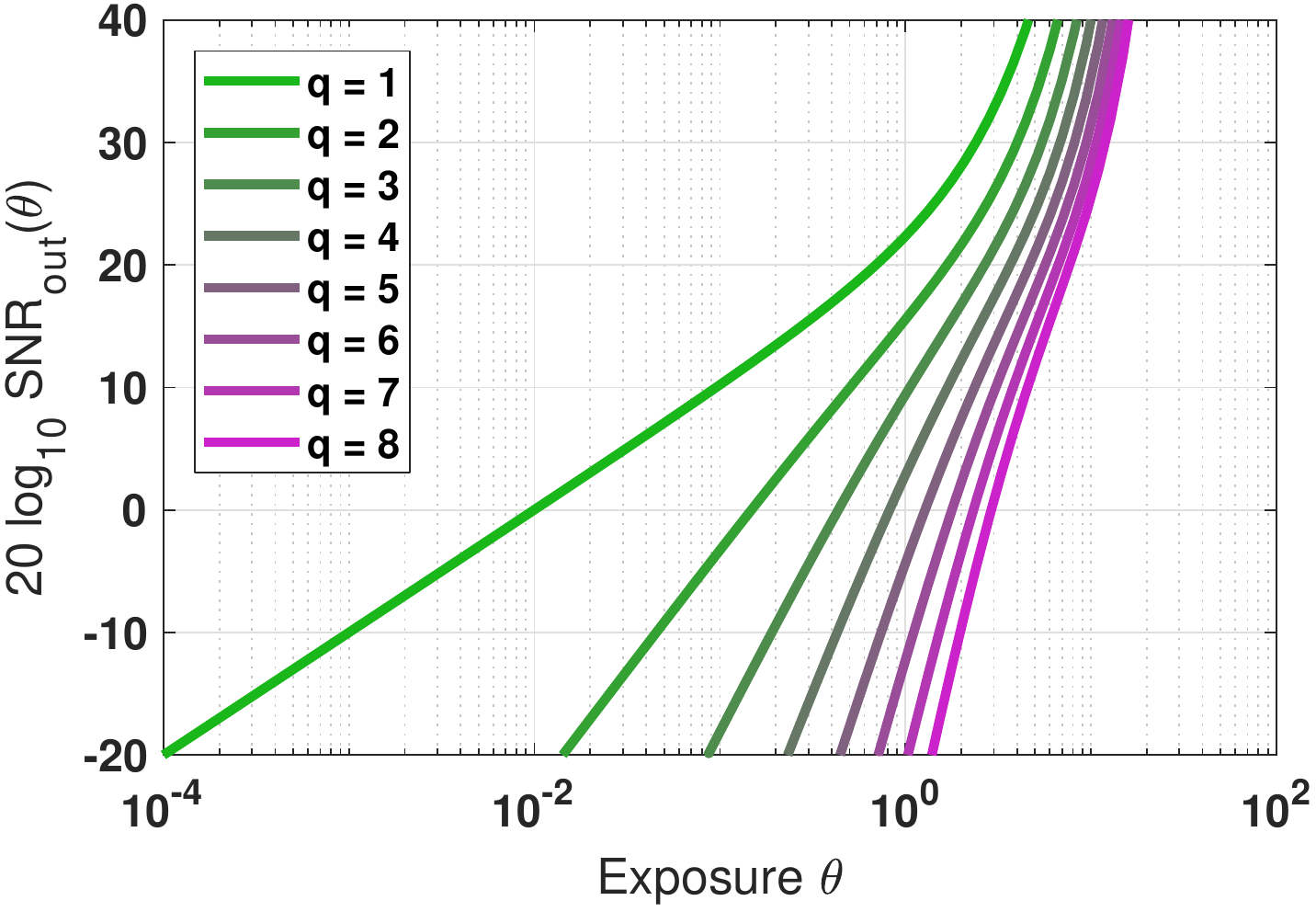}
\vspace{-4ex}
\caption{Output-referred SNR for a one-bit quanta image sensor, assuming that $\sigma_{\text{read}} = 0$. As $\theta$ goes beyond the threshold $q$, $\text{SNR}_{\text{out}}(\theta)$ continues to grow because of the inability of $\text{SNR}_{\text{out}}(\theta)$ to handle pixel saturation.}
\label{fig: ch3 SNR output}
\end{figure}

\section{$\text{SNR}_{\text{exp}}(\theta)$ for Finite full-well Capacity}
The subject of this paper is to understand the exposure-referred SNR for digital (CCD and CMOS) image sensors with a finite full-well capacity $L$. In particular, the goal is to understand the situation when $L$ is small, e.g., a few bits. This section presents the main result for such a scenario.

\subsection{$\text{SNR}_{\text{exp}}(\theta)$ for Truncated Poisson}
To make the analysis tractable, the derivation in this section will be focusing on a truncated Poisson distribution assuming $\sigma_{\text{read}} = 0$. Extension to the more complex noise model will be analyzed later.

\begin{theorem}[$\text{SNR}_{\text{exp}}(\theta)$ for truncated Poisson]
\label{thm: Ch01-4 SNR clipped Poisson}
Let $Y_1,\ldots,Y_N$ be i.i.d. random variables following the truncated Poisson statistics defined in \eref{eq: truncated Poisson} where $X_n \overset{\text{\tiny i.i.d.}}{\sim} \text{Poisson}(\theta)$ for $n = 1,\ldots,N$. Let $\thetahat(\overline{Y})$ be an estimator satisfying the mean invariance property. Then the exposure-referred SNR is
\begin{equation}
\text{SNR}_{\text{exp}}(\theta) = \sqrt{N} \cdot \frac{\theta}{\sqrt{\Var[Y_n]}} \cdot \frac{d\mu}{d\theta},
\label{eq: SNR exp main}
\end{equation}
where
\begin{align}
\E[Y]    &= \theta \Psi_{L-1}(\theta) + L (1-\Psi_{L}(\theta)) \bydef \mu,\notag\\
\Var[Y]  &= \theta^2 \Psi_{L-2}(\theta) + \theta\Psi_{L-1}(\theta) + L^2(1-\Psi_L(\theta)) - \mu^2,\notag\\
\frac{d\mu}{d\theta} &= \theta \Psi_{L-1}'(\theta) + \Psi_{L-1}(\theta)- L\Psi_{L}'(\theta). \label{eq: SNR exp main components}
\end{align}
\end{theorem}
\begin{proof}
The proof is presented in the Appendix.
\end{proof}

To illustrate the predicted $\text{SNR}_{\text{exp}}(\theta)$ as a function of $\theta$, \fref{fig: SNR_truncated_Poisson} shows several curves evaluated at different full-well capacity $L$. As is consistent with the one-bit QIS example shown in Section III.D, the exposure-referred SNR for a truncated Poisson random variable also demonstrates a drop in $\text{SNR}_{\text{exp}}(\theta)$ after the pixel saturates. What is more interesting is that as $L$ increases, $\text{SNR}_{\text{exp}}(\theta)$ becomes a straight line in the log-log plot with a sharp decay after saturation.

\begin{figure}[ht]
\centering
\includegraphics[width=\linewidth]{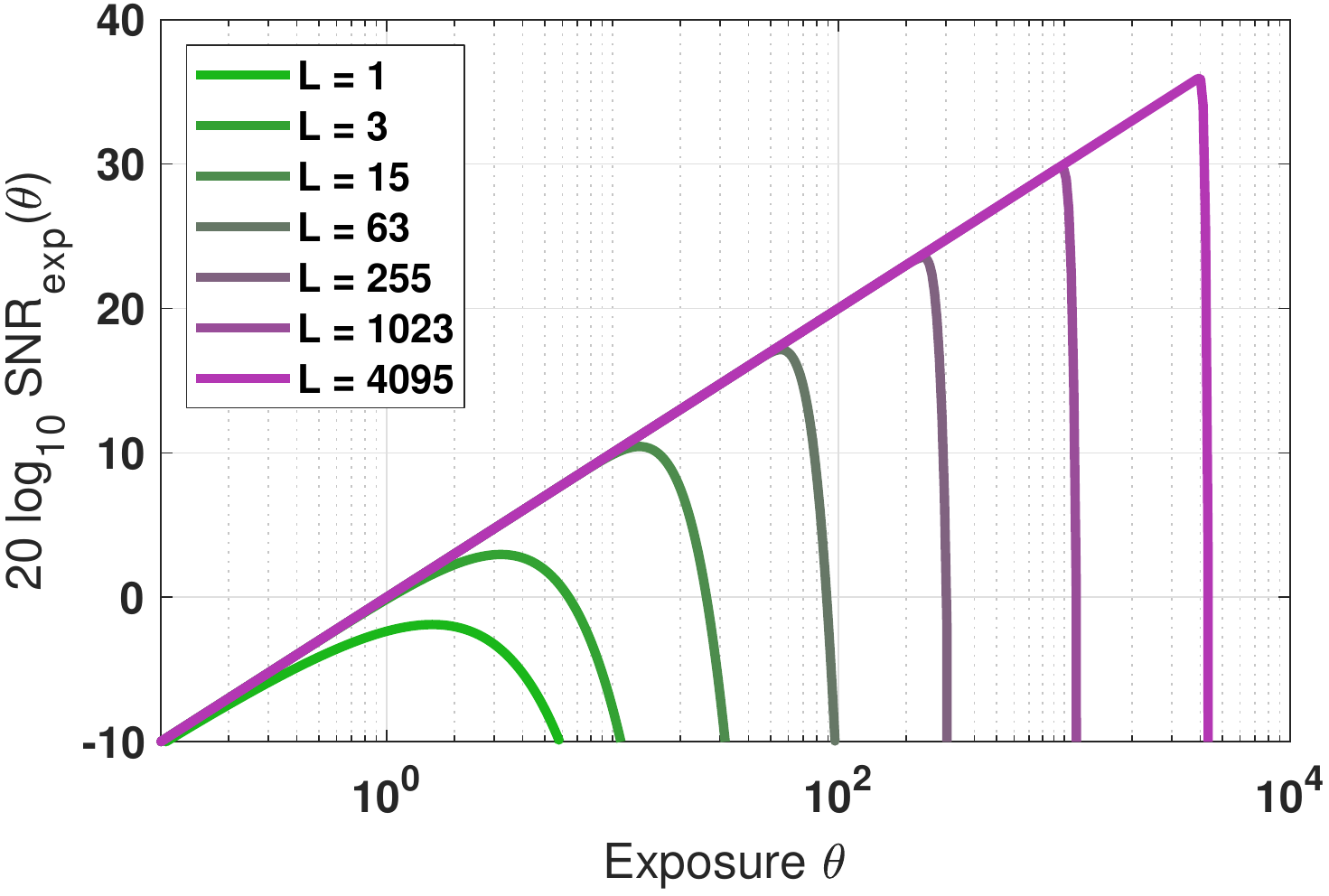}
\vspace{-4ex}
\caption{Exposure-referred SNR for a digital image sensor with a full-well capacity of $L$ electrons.}
\label{fig: SNR_truncated_Poisson}
\end{figure}

The rapid drop after the saturation is attributed to two reasons. First, as explained in the remark in Section II.C, the log-log plot has a compression of the $x$-axis so that the slope is amplified with $\theta$. If one plots the $x$-axis in the linear scale (instead of the log scale), the sharp cutoff will appear in a smoother transition. However, since in practice the exposure is always shown in the log scale, what is being shown in \fref{fig: SNR_truncated_Poisson} is valid. The second reason for the drop after the saturation is due to the limiting behavior of the incomplete Gamma function. As $L$ increases, the incomplete Gamma function in the log-log plot will have an increasingly sharp transient as shown in \fref{fig: incomplete gamma function}. This will be shown theoretically in the next subsection.

\subsection{Limiting Case}
\fref{fig: SNR_truncated_Poisson} shows that as the full-well capacity $L$ increases, $\text{SNR}_{\text{exp}}(\theta)$ becomes more linear in the log-log plot. Such a behavior can be theoretically derived by analyzing the limiting cases of the incomplete Gamma function. The log-log plot requires the $x$-axis to be mapped from $\theta$ to $\log_{10}\theta$. In this case, define $\phi = \log_{10}\theta$ (so that $\theta = 10^\phi$) and it can be shown that
\begin{align}
\Psi_L(10^\phi)
\approx
\begin{cases}
1, &\qquad \phi \le \log_{10} L,\\
0, &\qquad \phi > \log_{10} L,
\end{cases}
\label{eq: limit Psi}
\end{align}
assuming that $L \gg 1$.

\begin{corollary}
\label{corollary: limiting case}
Consider the same conditions as in Theorem~\ref{thm: Ch01-4 SNR clipped Poisson} but with $L \gg 1$. Let $\Psi = \Psi_L(\theta)$. Under the limiting assumption of $\Psi_L(\theta)$ described in \eref{eq: limit Psi}, it holds that
\begin{align}
20\log_{10} \text{SNR}_{\exp}(10^\phi)
=
\begin{cases}
10\phi,             &\qquad \phi \le \log_{10} L,\\
-\infty,            &\qquad \phi > \log_{10} L.
\end{cases}
\label{eq: HDR SNR}
\end{align}
\end{corollary}
\begin{proof}
See Appendix.
\end{proof}

The implication of the corollary is that as $L$ increases, plotting $\text{SNR}_{\text{exp}}(\theta)$ in the log-log plot will give a linear response followed by an abrupt transition. This is exactly what is happening in the output-referred SNR. Therefore, Theorem~\ref{thm: Ch01-4 SNR clipped Poisson} is a generalized version of the output-referred SNR curves reported in the literature. For practical algorithms such as those for high dynamic range imaging, \eref{eq: HDR SNR} is very common, for example used in \cite{granados2010optimal}.

\section{Monte Carlo Simulation}
\label{sec: Monte Carlo}
So far, the theoretical derivations have been focusing on the Poisson distribution only. Read noise, dark current, quantization error, and other sources of noise have not been considered. When including these factors, seeking an analytic expression would be significantly more challenging. A more reasonable approach is to resort to numerical schemes to estimate the SNR approximately.

\subsection{General Principle}
In general, the measurement $Y$ generated by an image sensor is the result of a sequence of optical-electronic operations such as
\begin{equation}
Y = \text{clip}\left\{\text{round}\left\{\text{Poisson}(\theta + \theta_{\text{dark}}) + \text{Gauss}(0,\sigma_{\text{read}}^2)\right\}\right\},
\label{eq: forward model}
\end{equation}
where $\theta_{\text{dark}}$ denotes the dark current, ``round'' denotes the analog-to-digital (A/D) conversion, and ``clip'' denotes the saturation due to a finite full-well capacity. Assuming a sufficiently large full-well capacity $L$, the output-referred SNR is given by \cite{granados2010optimal,Gamal_2005_magazine}
\begin{equation}
\text{SNR}_{\text{out}}(\theta) =
\begin{cases}
\frac{\theta}{\sqrt{\theta + \theta_{\text{dark}} + \sigma_{\text{read}}^2}}, &\quad \theta < L,\\
0, &\quad \theta \ge L,
\end{cases}
\label{eq: SNR out dark read}
\end{equation}
where $L$ is the full-well capacity.

To compute $\text{SNR}_{\text{exp}}(\theta)$, the numerical approach is to draw samples from the the distribution defined by the forward model:
\begin{equation}
Y_m = \text{forward model}\left(\theta \; | \; \theta_{\text{dark}}, \sigma_{\text{read}}, L\right),
\label{eq: Monte Carlo}
\end{equation}
for $m = 1,\ldots,M$, where $M$ denotes the number of Monte Carlo samples. As stated in \eref{eq: Monte Carlo}, the $m$th sample $Y_m$ is a function of the underlying signal $\theta$, along with other model parameters. Do not confuse $M$ with the number of i.i.d. measurements $N$ used in the previous subsections when defining the average $\overline{Y}$.

The Monte Carlo sampling scheme goes as follow. Consider $N = 1$. For every $\theta$, the sample average is an estimate of $\E[Y]$ and the sample variance is an estimate of $\Var[Y]$:
\begin{align*}
\widehat{\mu}(\theta) = \frac{1}{M}\sum_{m=1}^M Y_m, \;\;\;\text{and}\;\;\; \widehat{\sigma}^2(\theta) = \frac{1}{M}\sum_{m=1}^M (Y_m - \widehat{\mu})^2.
\end{align*}
Once $\widehat{\mu}(\theta)$ has been determined for every $\theta$, the derivative $d\mu/d\theta$ can be approximated by
\begin{align*}
\frac{d \widehat{\mu} }{d\theta}= \frac{\widehat{\mu}(\theta_{k+1})-\widehat{\mu}(\theta_{k})}{\theta_{k+1}-\theta_k},
\end{align*}
where $\{\theta_k \,|\, \theta_k < \theta_{k+1}, \; k = 1,\ldots,K\}$ is the discrete set of exposures used to evaluate the mean and variance. Consequently, if there are $N$ i.i.d. samples, $\text{SNR}_{\text{exp}}(\theta)$ can be approximately estimated by
\begin{equation}
\widehat{\text{SNR}}_{\text{exp}}(\theta) = \sqrt{N} \cdot \frac{\theta}{\widehat{\sigma}} \cdot \frac{d \widehat{\mu} }{d\theta}.
\end{equation}

%The MATLAB implementation of this Monte Carlo simulation scheme can be found in the Appendix.
%

\subsection{Visualizing the Impacts of $\theta_{\text{dark}}$ and $\sigma_{\text{read}}$}
With the Monte Carlo simulation technique, complex forward models can be visualized. Consider the following two demonstrations.

\begin{example}[Influence of Read Noise] \label{eg: influence of read noise}
The first scenario considers a fixed dark current, full-well capacity, and A/D converter, but a varying read noise level. Let $\theta_{\text{dark}} = 0.016$ (which is consistent with the quanta image sensor \cite{Ma:17}), a full-well capacity of $L = 15$ electrons, and 4-bit A/D converter. The read noise level $\sigma_{\text{read}}$ varies from 0 to 4.5 with a step interval of 0.5. By using $M = 5\times10^{6}$ Monte Carlo samples, the numerically simulated $\text{SNR}_{\text{exp}}(\theta)$ is plotted in \fref{fig: SNR_read}.

\begin{figure}[ht]
\centering
\includegraphics[width=\linewidth]{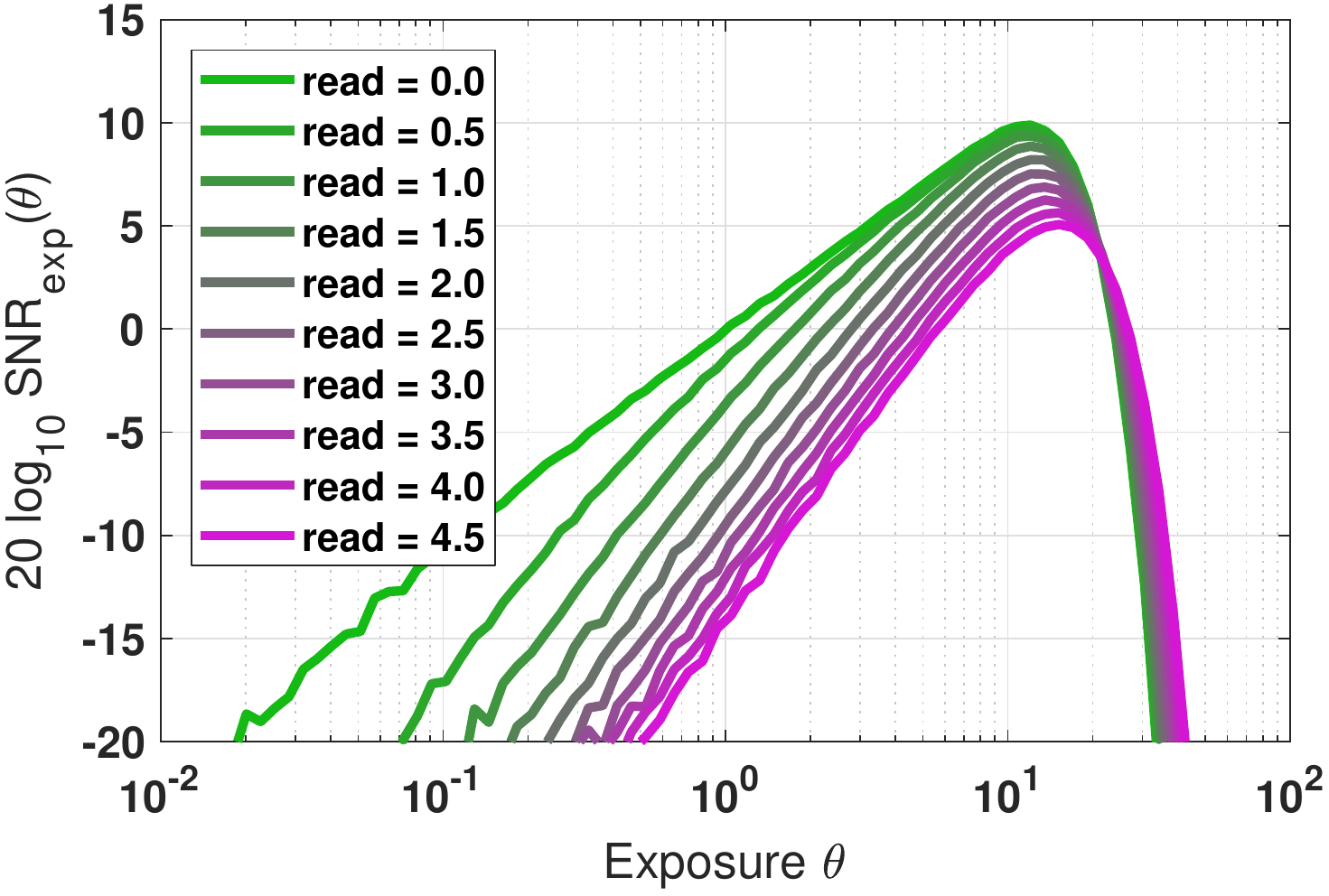}
\vspace{-4ex}
\caption{Exposure-referred SNR for a digital image sensor by considering different levels of read noise. }
\label{fig: SNR_read}
\end{figure}

Increasing the read noise leads to a reduced SNR for all $\theta$ before saturation. After saturation, the read noise will occasionally move a saturated measurement back to an unsaturated state because the Gaussian noise can take a negative value. See that the purple curves on the right-hand side of the plot are higher than the green curves. Therefore, for large $\theta$, there is a minor but noticeable gain in SNR, especially when the read noise is high. This is not necessarily a better outcome, because the increased SNR at larger $\theta$ comes at the cost of significantly lower SNR in low light where the $\theta$ is small.

Remark: The small fluctuation towards the tail in \fref{fig: SNR_read} is due to the randomness in the Monte Carlo simulation. As $M$ goes to infinity, the random estimate will approach the expectation by the law of large number. \hfill $\square$
\end{example}

\begin{example}[Influence of Dark Current]
The second scenario considers a fixed read noise, full-well capacity, and A/D converter, but a varying dark current. To be consistent with the literature, the dark current $\theta_{\text{dark}}$ is assumed to vary from 0 to 0.45 with a step interval of 0.05. The read noise level is fixed at 0.2 based on \cite{Ma:17}. The full-well capacity is $L = 15$ electrons, and a 4-bit A/D converter is used. Same as Example \ref{eg: influence of read noise}, $M = 5\times10^6$ Monte Carlo samples are used to numerically generate the $\text{SNR}_{\text{exp}}(\theta)$ plot in \fref{fig: SNR_dark}.

Unlike Example \ref{eg: influence of read noise} where the read noise has a substantial influence to the SNR, an increased dark current will only show its impact for small $\theta$. This should not be a surprise because when the true signal $\theta$ is strong, the influence of $\theta_{\text{dark}}$ will be negligible considering the small magnitude it usually has. For small $\theta$, the impact of $\theta_{\text{dark}}$ is more prominent. A smaller dark current indeed leads to a higher SNR as expected. \hfill $\square$

\begin{figure}[ht]
\centering
\includegraphics[width=\linewidth]{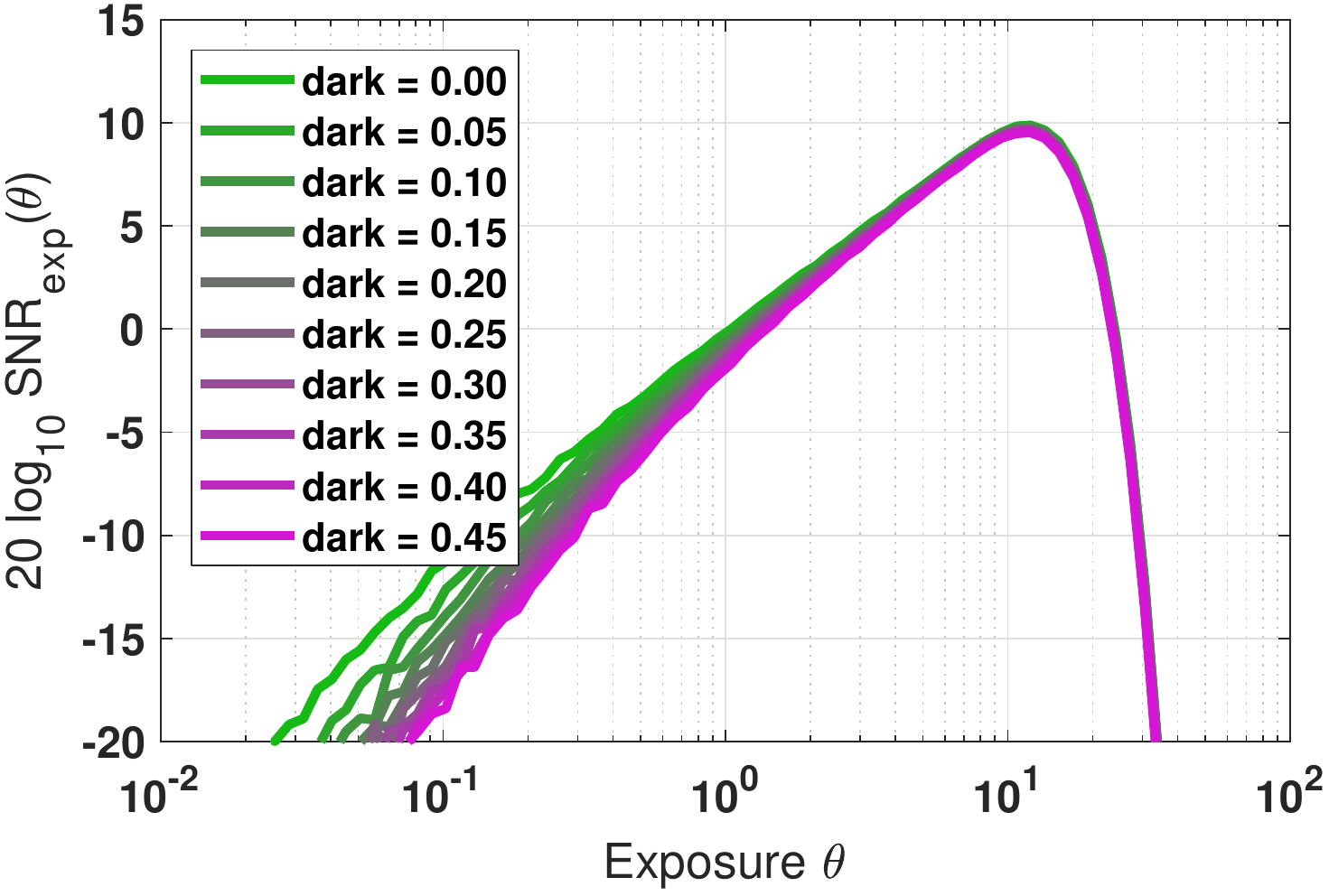}
\vspace{-4ex}
\caption{Exposure-referred SNR for a digital image sensor by considering different levels of dark current. The small fluctuation towards the tail on the left-hand side is due to randomness in the Monte Carlo simulation. }
\label{fig: SNR_dark}
\end{figure}
\end{example}

The utility of the Monte Carlo simulation is that it bypasses the complication of seeking for an analytic expression of $\text{SNR}_{\text{exp}}(\theta)$. To account for even more difficult modelings such as the pixel response non-uniformity, $1/f$ noise, conversion gain, and exposure time, etc, one just needs to modify the forward image formation model. For extreme cases such as very small $\theta$ where the random fluctuation is significant, one easy fix is to approximate $\text{SNR}_{\text{exp}}(\theta)$ by $\text{SNR}_{\text{out}}(\theta)$ using \eref{eq: SNR out dark read}. This approximation is reasonably accurate for small $\theta$ that is sufficiently far away from the saturation cutoff.

\section{Utilities}

After elaborating on the details of the exposure-referred SNR, readers may ask: what are the utilities of this SNR? The answer is simple. As a performance metric of an image sensor, the primary utility of the exposure-referred SNR is to describe how well an image sensor performs. Because of this primary goal, three points should be noted:
\vspace{2.5ex}
\begin{itemize}
\item The exposure-referred SNR is a generalized version of the output-referred SNR. The latter is a special case of the former when the bit-depth $L$ is large as shown in \fref{fig: figure08 exp vs out}. The output-referred SNR and the exposure-referred SNR are very similar for large $L$, whereas, for a small bit-depth, the output-referred SNR cannot capture the phenomenon when the exposure goes beyond the full-well capacity.
\item Because of the first point, any subsequent low bit-depth sensors applications need to use the exposure-referred SNR. Using the output-referred SNR will lead to sub-optimal performance, and this will be illustrated in Application 1 on high-dynamic range imaging.
\begin{figure}[th]
\centering
\includegraphics[width=\linewidth]{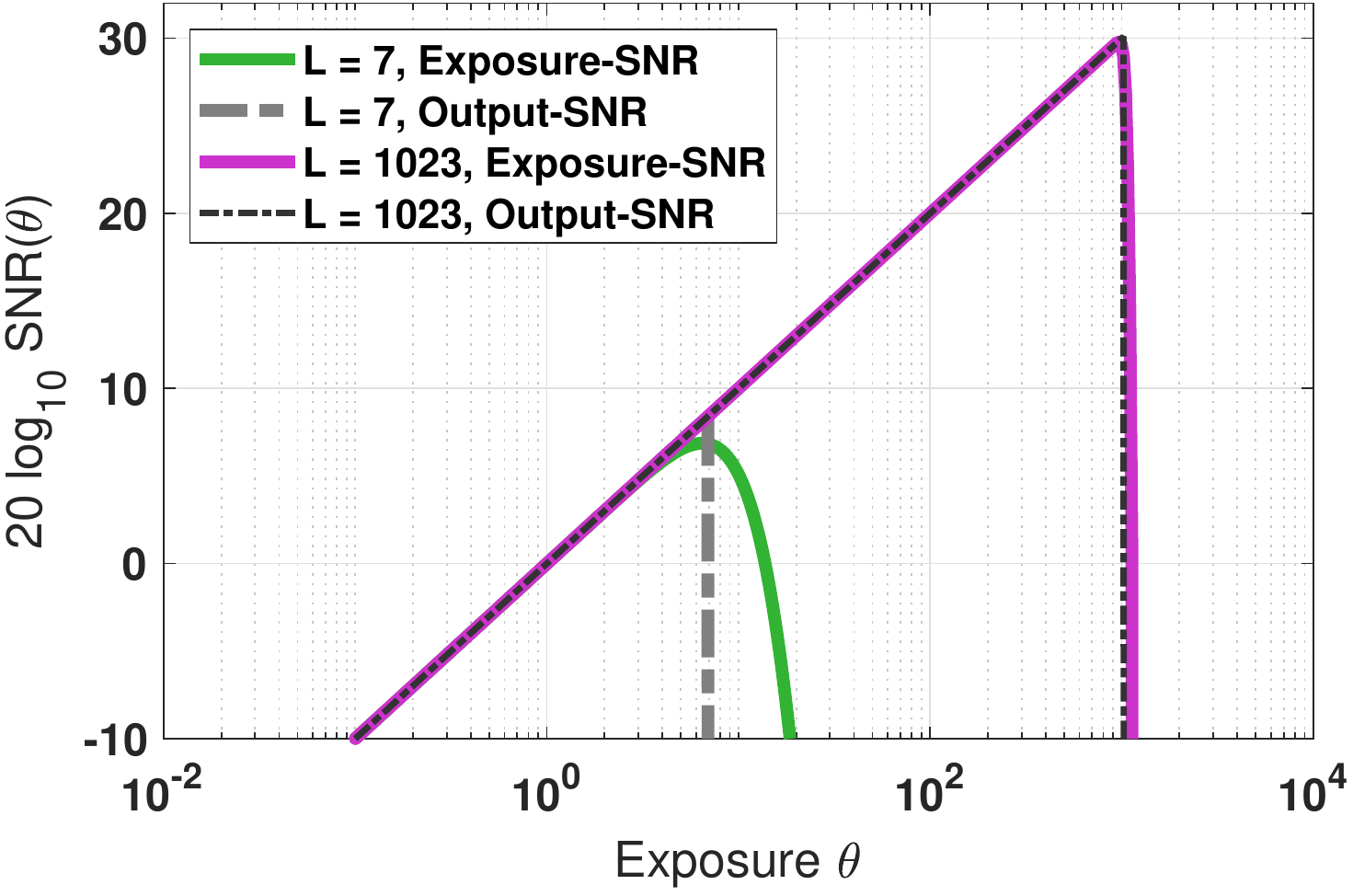}
\vspace{-4ex}
\caption{The output-referred SNR is a special case of the exposure-referred SNR when the full-well capacity $L$ is large. This figure plots the SNRs of a truncated Poisson random variable. For the output-referred SNR, the curve is generated by $\text{SNR}_{\text{out}}(\theta) = \sqrt{\theta}$ for $\theta \le L$ and $\text{SNR}_{\text{out}}(\theta) = 0$ for $\theta > L$.}
\label{fig: figure08 exp vs out}
\end{figure}
\item Again, because of the first point, the exposure-referred SNR can be used as the objective function to tune the parameters of the one-bit and few-bit image sensors. Since closed-form expressions of the SNR are available for some cases, the optimal parameters can also be expressed in closed-form. Application 2 on the threshold design will illustrate this point.
\end{itemize}

\subsection{High Dynamic Range Imaging Using a 3-bit Sensor}
The majority of the high dynamic range (HDR) image reconstruction algorithms today are designed for digital image sensors with a large full-well capacity, and so $\text{SNR}_{\text{out}}(\theta)$ is adequate. However, for image sensors with a small full-well capacity, e.g., $L = 7$, a fusion algorithm based on $\text{SNR}_{\text{exp}}(\theta)$ will produce a better image.

To illustrate the impact of $\text{SNR}_{\text{exp}}(\theta)$ in HDR imaging, consider the problem of reconstructing one HDR image (a single pixel) $\thetahat$ from $M$ exposure brackets $\mY^{1},\ldots,\mY^{M}$. Each of these exposure brackets $\mY^{m} = [Y_1^{m},\ldots,Y_N^{m}]$ has $N$ frames. The formation of each frame follows the equation
\begin{equation*}
Y^{m}_n \overset{\text{\tiny i.i.d.}}{\sim}
\begin{cases}
\text{Poisson}(\tau_m \theta), &\quad \tau_m \theta < L,\\
L,                             &\quad \tau_m \theta \ge L,
\end{cases}
\end{equation*}
for $n = 1,\ldots,N$, where $\tau_m$ is the $m$th integration time. To reconstruct the $m$th low-dynamic range (LDR) image, one can substitute the sample average $\overline{Y}^{m} = (1/N)\sum_{n=1}^N Y_n^m$ into the inverse mapping of the mean $\mu^{-1}$. Then the HDR image is fused by a linear combination of the estimates \footnote{The problem here does not assume any motion so that the theoretically optimal solution can be analytically derived. Handling motion remains an open problem in computer vision, although significant progress has been made over the past decade.}
\begin{equation}
\thetahat = \sum_{m=1}^M w_m \frac{\mu^{-1}(\overline{Y}^m)}{\tau_m},
\end{equation}
where $\{w_1,\ldots,w_M\}$ is a set of linear combination weights. The development of the idea can be traced back to Mann and Picard \cite{mann94b}, Debevec and Malik \cite{Debevec_Malik_HDR}, among other works \cite{nayar2000high, granados2010optimal, kirk2006noise}. For a conventional CIS with a large full-well capacity, $\mu^{-1}$ is an identity mapping and so the reconstruction is simplified to $\thetahat = \sum_{m=1}^M \frac{w_m}{\tau_m} \overline{Y}^m$. For sensors with a small full-well capacity, the nonlinearity of the mean function should be taken into account.
This reconstruction scheme is elementary and predates all the deep learning methods.

As previously proved by Gnanasambandam and Chan in \cite{Gnanasambandam_TCI_HDR}, the optimal weight $w_m$ is
\begin{equation}
w_m = \frac{\text{SNR}_m^2(\theta)}{\sum_{m=1}^M \text{SNR}_m^2(\theta)},
\end{equation}
where $\text{SNR}_m(\theta)$ is the SNR of the $m$th low dynamic range image. If $\text{SNR}_{\text{out}}(\theta)$ is used, the weight will become
\begin{align}
w_m
= \frac{\text{SNR}_{\text{out},m}^2(\theta)}{\sum_{m=1}^M \text{SNR}_{\text{out},m}^2(\theta)}
&= \frac{(\sqrt{\tau_m \theta \cdot \mathbb{I}_{\{\tau_m\theta<L\}}})^2}{\sum_{m=1}^M (\sqrt{\tau_m \theta \cdot \mathbb{I}_{\{\tau_m\theta<L\}}})^2} \notag\\
&= \frac{\tau_m \cdot \mathbb{I}_{\{\tau_m\theta<L\}}}{\sum_{m=1}^M \tau_m \cdot \mathbb{I}_{\{\tau_m\theta<L\}}}, \label{eq: optimal weight output}
\end{align}
where the function $\mathbb{I}_{\{\tau_m\theta<L\}}$ is a binary indicator showing whether the total exposure $\tau_m\theta$ has exceeded the full-well capacity. $\mathbb{I}_{\{\tau_m\theta<L\}} = 1$ if the argument is true and $\mathbb{I}_{\{\tau_m\theta<L\}} = 0$ if the argument is false. If $\text{SNR}_{\text{exp},m}(\theta)$ is used, the optimal weight will become
\begin{equation}
w_m = \frac{\text{SNR}_{\text{exp},m}^2(\theta)}{\sum_{m=1}^M \text{SNR}_{\text{exp},m}^2(\theta)},
\end{equation}
where one can substitute \eref{eq: SNR exp main} into this equation.

Once the combined image is formed, the overall SNR can be computed via
\begin{equation}
\text{SNR}_{\text{HDR}}(\theta) = \frac{\theta}{\sqrt{\sum_{m=1}^M \left(\frac{w_m}{\tau_m}\right)^2 \sigma_{m}^2(\theta)}}.
\label{eq: SNR HDR}
\end{equation}
where $\sigma_{m}^2(\theta) = \Var[\mu^{-1}(\overline{Y}^m)]$ is the noise variance of the $m$th reconstructed LDR image. The intuitive interpretation of \eref{eq: SNR HDR} is that it weighs the noise according to the integration time $\tau_m$ and combination weight $w_m$ to produce a calibrated noise. Thus, the ratio is the SNR with respect to the optimally combined image $\thetahat$.

The question to be answered here is: If the weight $w_m$ is computed by using the output-referred SNR while the actual sensor has a small full-well capacity such as $L = 7$, what will $\text{SNR}_{\text{HDR}}(\theta)$ be? To answer this question, consider the following configurations: Assume four integration times $\tau_1 = 1$, $\tau_2 = 0.1$, $\tau_3 = 0.01$, $\tau_4 = 0.001$, a total number of frames $N = 100$ at each integration time, and a full-well capacity of $L = 7$. The output-referred SNR for each $m$, $\text{SNR}_{\text{out},m}$, is computed by \eref{eq: optimal weight output}, whereas the exposure-referred SNR for each $m$, $\text{SNR}_{\text{exp},m}$, is computed by \eref{eq: SNR exp main}. Once computed, the weight $w_m$ is constructed from $\text{SNR}_{\text{out},m}$. The overall $\text{SNR}_{\text{HDR}}(\theta)$ is formed by using \eref{eq: SNR HDR} where the calibrated noise $\sigma_{m}^2(\theta)$ uses the exposure-referred noise $\sigma_{\text{exp},m}^2(\theta)$, essentially $\Var[Y] \cdot \frac{d\mu}{d\theta}$ defined in \eref{eq: SNR exp main}.

\begin{figure}[th]
\centering
\includegraphics[width=\linewidth]{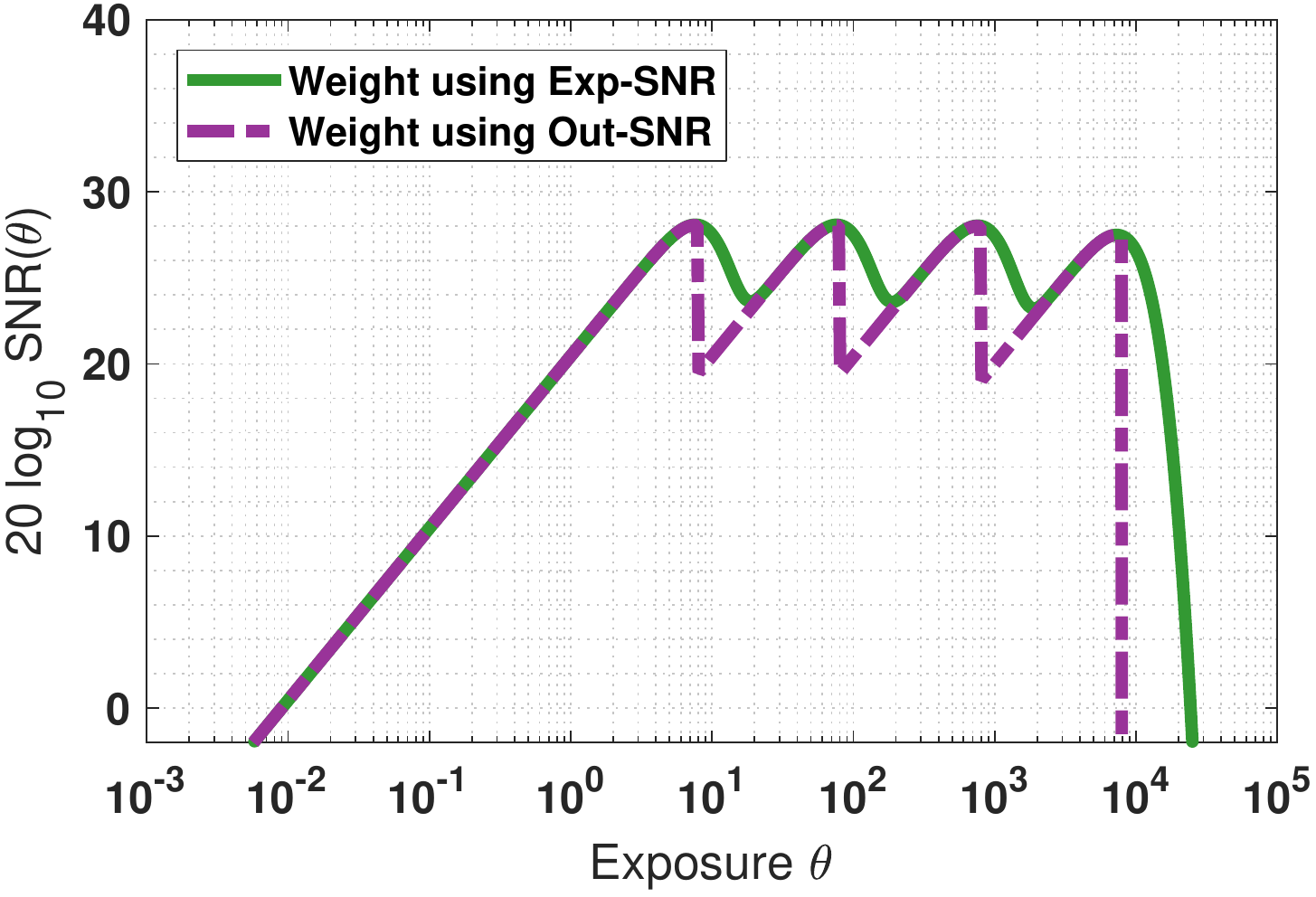}
\vspace{-4ex}
\caption{Comparison of the exposure bracketing by using two different schemes to compute the weight when the sensor has a full-well capacity of $L = 7$. (i) Use the exposure-referred SNR to compute the weight. (ii) Use the output-referred SNR to compute the weight. }
\label{fig: figure09 bracketing}
\end{figure}

Since the weight $w_m$ is computed using $\text{SNR}_{\text{out}}$ while the actual sensor has a small full-well capacity, the overall $\text{SNR}_{\text{HDR}}$ will suffer. \fref{fig: figure09 bracketing} shows the $\text{SNR}_{\text{HDR}}$'s where the weights are either computed from $\text{SNR}_{\text{out}}$ or $\text{SNR}_{\text{exp}}$. This is a new plot that has never been shown including \cite{Gnanasambandam_TCI_HDR}. As one can see, $\text{SNR}_{\text{HDR}}$ suffers in two places. The first place is the gap between two consecutive maxima, where the weights generated by the $\text{SNR}_{\text{out}}$ gives an overall $\text{SNR}_{\text{HDR}}$ with a sharp cutoff. This discontinuity will be visible when the scene contains a continuous range of $\theta$. The second place is the sharp cutoff after the full-well capacity. The visual impact is that the overall HDR image will saturate sooner than what it is supposed to be, compared to the fusion using exposure-referred SNR.

\fref{fig: HDR example} provides a visual comparison between the HDR image reconstructed using $\text{SNR}_{\text{out}}$ and $\text{SNR}_{\text{exp}}$. In this example, the full-well capacity is assumed to be $L = 7$. Four different exposures (1 ms, 0.1 ms, 0.01 ms, and 0.001 ms) were used to construct the low dynamic range images, where each image is the result of $N = 100$ frames averaged over time. The scene is static, and so the reconstructed results are theoretically optimal with respect to the linear combination and the choice of the SNR. The visual comparison shows a clear benefit of $\text{SNR}_{\text{exp}}$ over $\text{SNR}_{\text{out}}$ especially around the cropped areas where the pixels are near saturation. However, the performance gap will become smaller when the full-well capacity becomes larger.

\begin{figure}[ht]
\centering
\begin{tabular}{cc}
\hspace{-1.0ex}\includegraphics[width=0.47\linewidth]{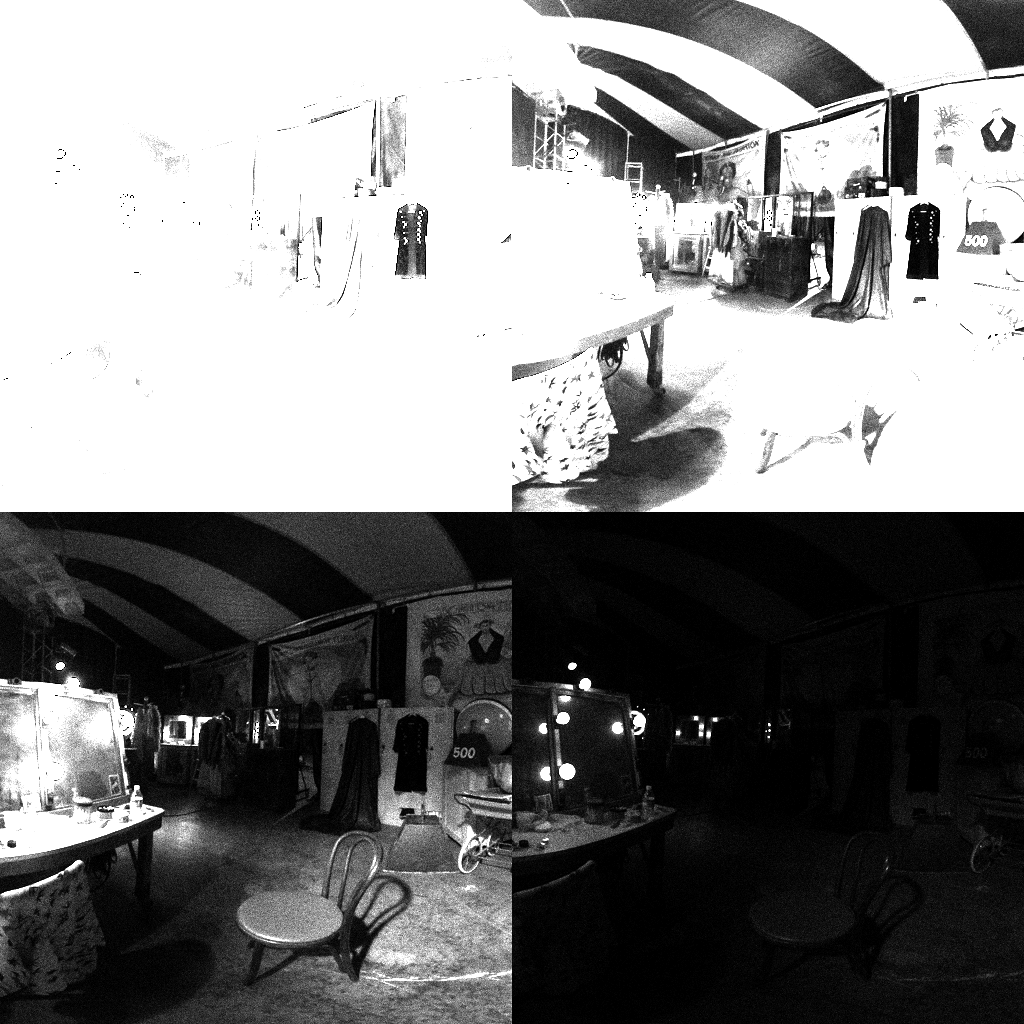} &
\hspace{-2.0ex}\includegraphics[width=0.47\linewidth]{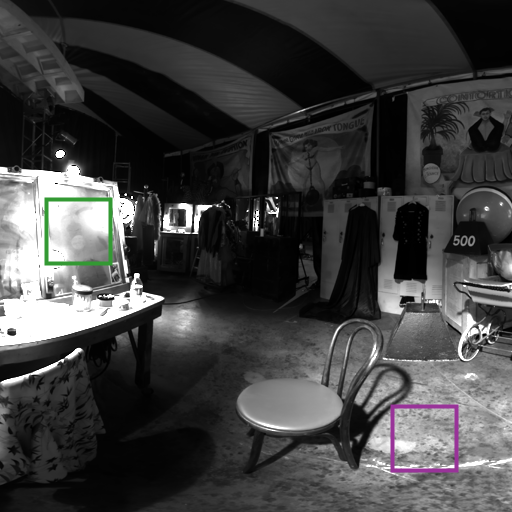} \\
(a) Exposures & (b) Ground Truth \\
\hspace{-1.0ex}\includegraphics[width=0.47\linewidth]{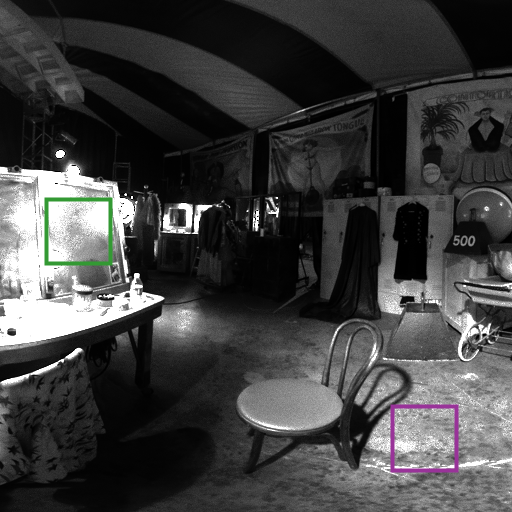}&
\hspace{-2.0ex}\includegraphics[width=0.47\linewidth]{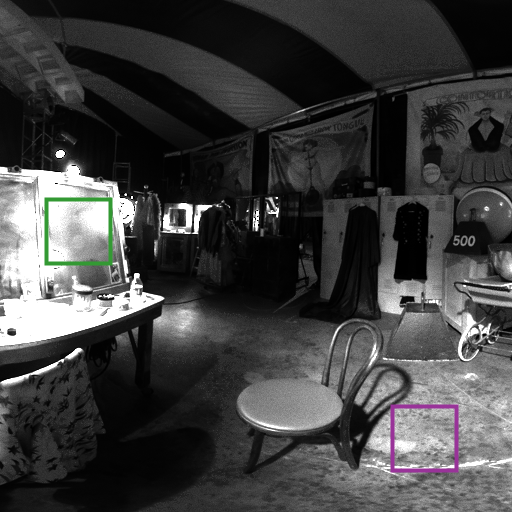} \\
(c) Using $\text{SNR}_{\text{out}}$ & (d) Using $\text{SNR}_{\text{exp}}$\\
PSNR = 37.53dB & PSNR = 45.02dB\\
\end{tabular}
\begin{tabular}{ccc}
\hspace{-1.0ex}\includegraphics[width=0.32\linewidth]{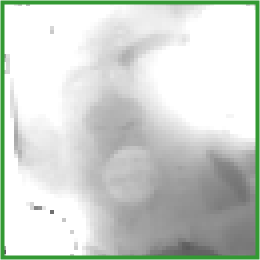} &
\hspace{-2.0ex}\includegraphics[width=0.32\linewidth]{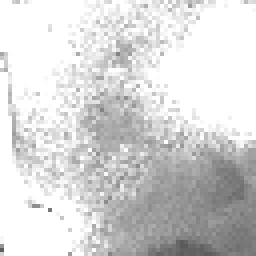} &
\hspace{-2.0ex}\includegraphics[width=0.32\linewidth]{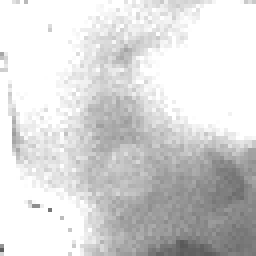} \\
\hspace{-1.0ex}\includegraphics[width=0.32\linewidth]{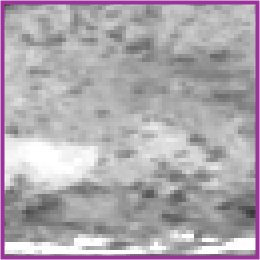} &
\hspace{-2.0ex}\includegraphics[width=0.32\linewidth]{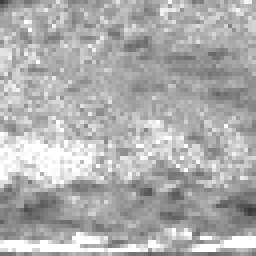} &
\hspace{-2.0ex}\includegraphics[width=0.32\linewidth]{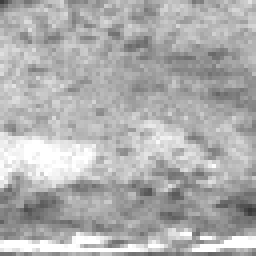} \\
(e) Ground Truth & (f) Using $\text{SNR}_{\text{out}}$ & (g) Using $\text{SNR}_{\text{exp}}$
\end{tabular}
\caption{High dynamic range (HDR) image reconstruction using the optimal linear combination scheme. (a) The four exposures are captured at 1 ms, 0.1 ms, 0.01 ms, and 0.001 ms. Shown in this sub-figure are the single-frame measurements with a full-well capacity of 3 bits. (b) The ground truth HDR image. (c) Linear combination using the output-referred SNR. (d) Linear combination using the exposure-referred SNR. For (c)-(d), the reconstruction is based on averaging 100 frames per exposure. (e)-(f)-(g) shows the zoom-in regions of the respective figures.}
\label{fig: HDR example}
\end{figure}

\subsection{Threshold for one-bit Quanta Image Sensors}
The second application is the theoretical analysis of the one-bit quanta images sensor (QIS). The one-bit QIS has an interesting forward model given by \eref{eq: QIS model} where the corresponding $\text{SNR}_{\text{exp}}(\theta)$ is derived in \eref{eq: SNR exp QIS} (assuming that the read noise $\sigma_{\text{read}}$ is negligible compared to the signal $\theta$.) This section will show two new results on the optimal threshold $q$ that were not mentioned in \cite{elgendy2018optimal}.

\subsubsection{Optimal $q$ for $\theta \gg 1$}  The context of this operating regime is that the sensor sees a sufficient amount of photons but it chooses to operate in a single-bit mode to earn the speed. In this case, the threshold $q$ should be dynamically adjusted to maximize the SNR.

When $\theta \gg 1$, the read noise $\sigma_{\text{\scriptsize read}}$ can be neglected. Consequently, the SNR follows the Poisson statistics as in \eref{eq: SNR exp QIS}. By noting that $0 \le \Psi_q(\theta) \le 1$ and hence $\Psi_q(\theta)(1-\Psi_q(\theta)) \le 1/4$ where the maximum is attained when $\Psi_q(\theta) = 1/2$, $\text{SNR}_{\text{exp}}(\theta)$ can be lower bounded by
\begin{align}
\text{SNR}_{\text{exp}}(\theta)
&= \frac{\theta}{\sqrt{\Psi_q(\theta)(1-\Psi_q(\theta))}} \cdot \frac{\theta^{q-1}e^{-\theta}}{(q-1)!} \notag \\
&\ge \frac{\theta}{\frac{1}{2}} \cdot \frac{\theta^{q-1}e^{-\theta}}{(q-1)!}. \label{eq: QIS SNR one-bit bound}
\end{align}
The following lemma shows the optimal $q$ for the lower bound.

\begin{lemma}[Maximizing the lower bound]
Consider the SNR lower bound for one-bit QIS given by \eref{eq: QIS SNR one-bit bound}. The bound is maximized when $q^* = \lfloor \theta \rfloor + 1$.
\end{lemma}
\begin{proof}
The proof of this lemma can be found in \cite{elgendy2018optimal}.
\end{proof}

What was not proved in \cite{elgendy2018optimal} is that at this optimal $q$, the inequality in \eref{eq: QIS SNR one-bit bound} is actually an equality. The reason is that for large $\theta$, $\Psi_q(\theta)$ can be approximated by the cumulative distribution function of a Gaussian via the Stirling's formula. This can be seen from \eref{eq: Psi approximate} where for large $\theta$,
\begin{equation*}
\Psi_q(\theta) \approx \int_{-\infty}^{q}\frac{1}{\sqrt{2\pi\theta}} \exp\left\{-\frac{(\theta+1-y)^2}{2\theta}\right\} dy.
\end{equation*}
The integral is $\Psi_q(\theta) = 1/2$ when $q = \theta+1$, because the integrand is a Gaussian probability distribution centered at $\theta+1$. In other words, at the optimal $q$, $\Psi_q(\theta) = 1/2$ and so $\Psi_q(\theta)(1-\Psi_q(\theta)) = 1/4$. Hence, the equality in \eref{eq: QIS SNR one-bit bound} is satisfied, meaning that $q = \theta+1$ does not only maximize the lower bound but it also maximizes the SNR.

\subsubsection{Optimal $q$ for $\theta \approx 1$}. This is a new result. The context is that the photon flux is small and so the goal of the sensor is to count the number of photons. However, for small $\theta$, the read noise plays a role because if $\sigma_{\text{read}}$ is big, two adjacent counts cannot be differentiated. The threshold $q$ in this context is used to quantize the analog voltage (which is a Poisson-Gaussian random variable) so that the SNR is maximized.

For one-bit Poisson-Gaussian, the probability distribution of the binary random variable $Y$ still follows \eref{eq: QIS model} but with $X \sim \text{Poisson}(\theta) + \text{Gaussian}(0,\sigma_{\text{\scriptsize read}}^2)$. The exposure-referred SNR is computable because $Y$ is binary. To see this, notice that the mean of $Y$ is
\begin{align}
\underset{\mu(\theta)}{\underbrace{\E[Y]}}
&= \sum_{k=0}^{\infty} \frac{\theta^k e^{-\theta}}{k!} \int_{-\infty}^{q} \frac{1}{\sqrt{2\pi\sigma_{\text{\scriptsize read}}^2}} \exp\left\{-\frac{(y-k)^2}{2\sigma_{\text{\scriptsize read}}^2}\right\} dy \notag \\
&= \frac{1}{2}\sum_{k=0}^{\infty} \frac{\theta^k e^{-\theta}}{k!} \text{erfc}\left(\frac{q-k}{\sqrt{2}\sigma_{\text{\scriptsize read}}}\right), \label{eq: mu, one-bit QIS 1}
\end{align}
where $\text{erfc}(\cdot)$ is the error function. In the first equation the term inside the summation is the convolution of a Poisson probability density and a Gaussian probability density, evaluated at the threshold $q$. \eref{eq: mu, one-bit QIS 1} can be computed using numerical techniques.

The variance of $Y$ follows from the fact that $Y$ is binary (i.e. a Bernoulli random variable). Thus,
\begin{align}
\Var[Y] = \E[Y](1-\E[Y]) = \mu(\theta)(1-\mu(\theta)). \label{eq: mu, one-bit QIS 2}
\end{align}
The derivative $d\mu/d\theta$ can be approximated numerically:
\begin{align}
\frac{d\mu}{d\theta} \approx \frac{\mu(\theta+\epsilon) - \mu(\theta)}{\epsilon}, \label{eq: mu, one-bit QIS 3}
\end{align}
where $\epsilon$ is a small numerical constant. Combining everything in \eref{eq: mu, one-bit QIS 1}, \eref{eq: mu, one-bit QIS 2} and \eref{eq: mu, one-bit QIS 3}, the exposure-referred SNR can be numerically computed via Theorem~\ref{thm: SNR exp}.

\begin{figure}[ht]
\centering
\includegraphics[width=\linewidth]{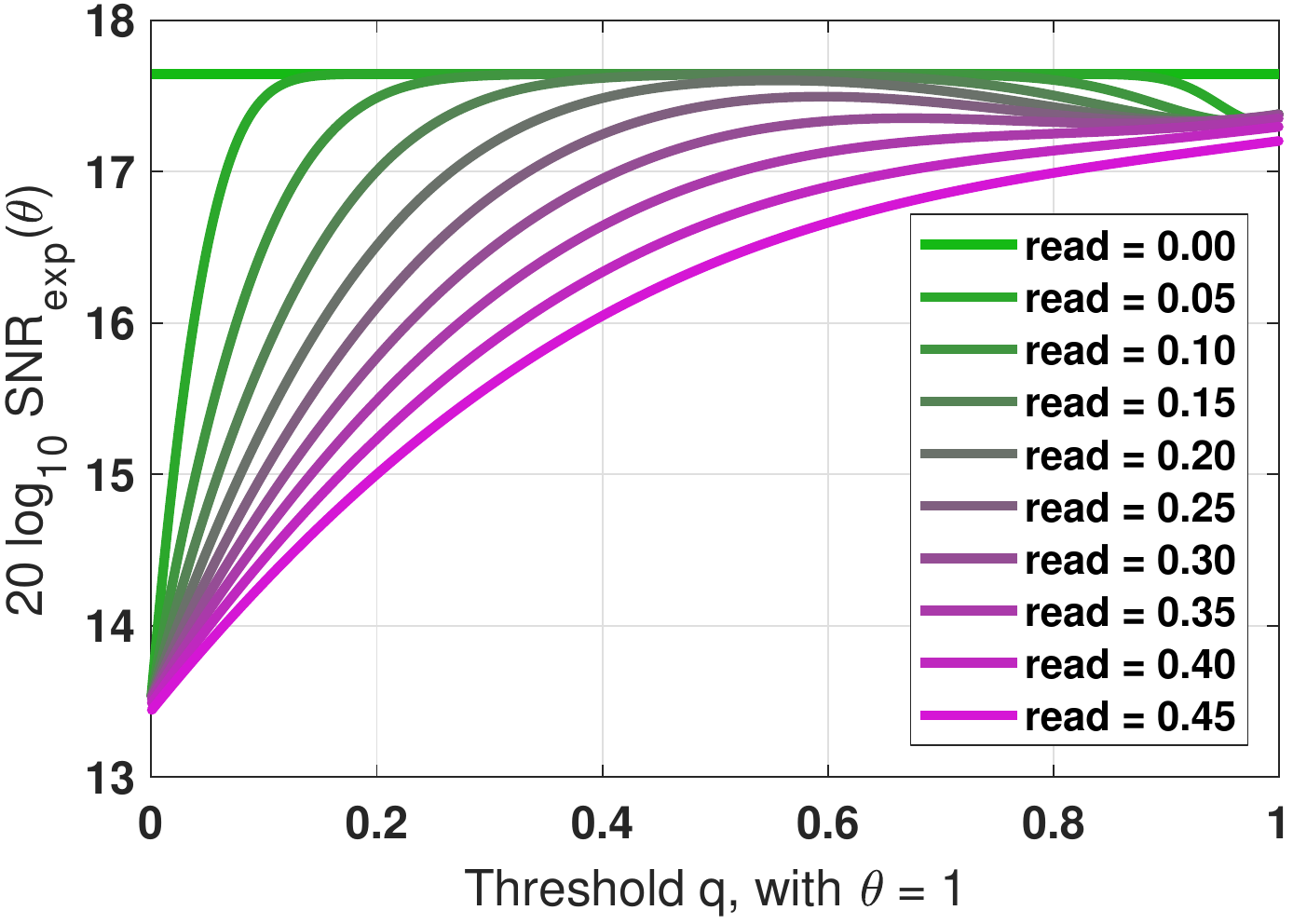}
\vspace{-4ex}
\caption{$\text{SNR}_{\text{exp}}$ as a function of the threshold when the exposure is $\theta = 1$. As the read noise $\sigma_{\text{read}}$, increases, the SNR drops when $q = 0.5$.}
\label{fig: threshold read noise}
\end{figure}

The derived $\text{SNR}_{\text{exp}}(\theta)$ is plotted in \fref{fig: threshold read noise} as a function of the threshold $q$, for various read noise levels $\sigma_{\text{read}}$. When $\sigma_{\text{read}} = 0$, the SNR is a constant for all $q$ considered. This is because in the absence of read noise, there is no ambiguity in determining the measured number of photons regardless where the threshold $q$ is put. As $\sigma_{\text{read}}$ increases, two adjacent counts begin to overlap due to the Gaussian. The maximum is located around $q = 0.5$. For high read noise where the two peaks of the photon count merge, there is limited SNR one can expect from the sensor.

\section{Discussions and Conclusion}
\subsection{Alternatives to SNR?}
While SNR is a natural choice for analyzing the performance of an image sensor, it is by no means the \emph{only} option. Especially for one-bit devices such as the quanta image sensor, there are other ways to characterize the performance.

\subsubsection{Entropy}
As far as one-bit measurements are concerned, the entropy is a natural substitute of the SNR. If $Y$ is binary with $p_Y(1) = 1-\Psi_q(\theta)$ and $p_Y(0) = \Psi_q(\theta)$, the entropy is
\begin{align*}
H(Y)
&= - p_Y(1) \log_2 p_Y(1) - p_Y(0)\log_2 p_Y(0) \notag \\
&= - (1-\Psi_q(\theta)) \log_2 (1-\Psi_q(\theta)) - \Psi_q(\theta)\log_2 \Psi_q(\theta).
\end{align*}
It is relatively easy to show that the derivative of the entropy with respect to $\Psi_q(\theta)$ is
\begin{align*}
\frac{d}{d\Psi_q(\theta)} H(Y) = -\log\left(\frac{1-\Psi_q(\theta)}{\Psi_q(\theta)}\right).
\end{align*}
Setting it to zero will yield $\Psi_q(\theta) = \frac{1}{2}$. Therefore, the entropy is maximized when $\E[Y] = 1-\Psi_q(\theta) = \frac{1}{2}$. Since $\E[Y]$ is the expected value of the measurement, $\E[Y] = \frac{1}{2}$ means that the entropy is maximized when there are 50\% one's and 50\% zero's in a set of independent measurements. So, if the application goal is to identify a threshold $q$ such that the performance of the sensor is maximized, then instead of optimizing for the SNR as in Section VI.B.1, the alternative is to optimize the entropy. The solution to $\Psi_q(\theta) = 1/2$ is $q^* = \theta + 1$, which is consistent with Section VI.B.1.

\subsubsection{Bit Error Rate (BER)}
In the presence of read noise, the bit error rate is another commonly used criterion to evaluate the performance of a sensor. For one-bit quanta image sensor, the BER measures the probability of making a wrong decision (i.e., declaring a 0 as a 1, or declaring a 1 as a 0). It can be readily computed as
\begin{align}
\text{BER}(\theta)
&= p_Y(0) \cdot \int_{q}^{\infty} \frac{1}{\sqrt{2\pi\sigma_{\text{read}}^2}}e^{-\frac{t^2}{2\sigma_{\text{read}}^2}}  dt \notag \\
&\quad + p_Y(1) \cdot \int_{-\infty}^{q} \frac{1}{\sqrt{2\pi\sigma_{\text{read}}^2}}e^{-\frac{(t-1)^2}{2\sigma_{\text{read}}^2}}  dt \notag\\
&=  \frac{1}{2}\text{erfc}\left(\frac{q}{\sigma_{\text{read}}\sqrt{2}}\right) \Psi_q(\theta) \notag \\
&\quad + \frac{1}{2}\text{erfc}\left(\frac{1-q}{\sigma_{\text{read}}\sqrt{2}}\right)(1-\Psi_q(\theta)). \label{eq: BER 2}
\end{align}
Therefore, if $q = 1/2$, the BER is simplified to
\begin{equation}
\text{BER}(\theta) = \frac{1}{2}\text{erfc}\left(\frac{1}{\sigma_{\text{read}} \sqrt{8}}\right),
\label{eq: BER}
\end{equation}
which does not depend on $\theta$. If $\text{BER}(\theta)$ can be empirically measured, then by inverting \eref{eq: BER} one can estimate the read noise $\sigma_{\text{read}}$. For a fixed $\theta$, one can also optimize \eref{eq: BER 2} by finding an appropriate $q$.

\subsection{Conclusion}
The exposure-referred SNR is a concept motivated by the need to capture the sensor's behavior near and beyond the full-well capacity. For small image sensors with one or few bits, exposure-referred SNR provides a natural characterization of the performance without showing an infinite SNR due to the artificial squeezing of the noise. In order to establish the exposure-referred SNR, the paper introduces new mathematical concepts and showed a few results:
\begin{itemize}
\item The mean invariance property is introduced. The property asserts that when an estimator $\thetahat(\cdot)$ is applied to the mean of the measurement, the mapped value is the true parameter $\thetahat(\mu) = \theta$.
\item Exposure-referred SNR calibrates the noise variance $\Var[Y]$ by the derivative $d\mu/d\theta$, where the derivative is the result of the first-order approximation used in the Delta Method.
\item SNR of a sensor with a finite full-well capacity is analytically derived via the incomplete Gamma function. The new result generalizes the conventional ones. For sensors with a large full-well capacity, the new result recovers the classical one that shows a linear response. For sensors with a small full-well capacity, the new result shows how the transient of the SNR looks like.
\item Monte Carlo simulation is presented for complex noise models. The general procedure is to estimate the mean of the measurement $\mu(\theta)$ and the variance of the measurement $\Var[Y]$. Then by using a numerical finite difference operator, the derivative $d\mu/d\theta$ can be approximated.
\item Optimal high dynamic range image reconstruction is shown. The new result generalizes the classical linear response.
\item Threshold analysis of one-bit sensors is revisited and generalized. For $\theta \gg 1$, the optimal threshold is $q = \theta + 1$. For $\theta \approx 1$, the optimal threshold needs to overcome the read noise and so the optimal value is $q = 1/2$.
\item For one-bit sensors, it was found that the entropy of the bits and the bit error rate can be used as alternatives to characterize the sensor.
\end{itemize}

As the full-well capacity of the image sensors is becoming small, it is anticipated that the exposure-referred SNR will become a useful utility for theoretical analysis of the sensors, hence allowing signal processing \cite{Gnanasambandam_Megapixel_2019,Elgendy_CFA_2021,Elgendy_demosaicking_2021,Chan_Density_2022}, algorithm development \cite{Chan_GlobalSIP_2014,Chan_PnP_2016,Chan_MDPI_2016,Choi_ICASSP_2018,chi_chan_ECCV2020}, and computer vision applications \cite{gnanasambandam_chan_ECCV2020, Chengxi}. Specific problems such as sensor gain control, exposure analysis, bit-depth to speed trade off, and color filter array, will be important questions to answer next.

\section*{Acknowledgement}
The work is supported, in part, by the United States National Science Foundation under the grants CCF-1718007, CCSS-2030570, IIS-2133032. The authors thank Professor Eric Fossum for discussing the ideas of this paper.

\section{Appendix}
\subsection{Gaussian approximation of Poisson}
When deriving the first-order derivative of the incomplete Gamma function, it was mentioned that the Poisson distribution can be approximated by a Gaussian. The formal statement is as follows.
\begin{lemma}[Gaussian approximation of Poisson]
For large $\theta$ (i.e., $\theta \gg 1$), it holds that
\begin{equation}
p_X(x) \bydef \frac{\theta^x e^{-\theta}}{x!} \approx \frac{1}{\sqrt{2\pi\theta}}e^{-\frac{(x-\theta)^2}{2\theta}}.
\end{equation}
\end{lemma}
Note that this is \emph{not} the Central Limit theorem because it does not involve any sample average. The approximation compares the two functions.

\begin{proof}
First of all, take the log on the Poisson equation:
\begin{equation*}
\log p_X(x) = \log \left\{\frac{\theta^x e^{-\theta}}{x!}\right\} = x \log \theta - \theta - \log x!
\end{equation*}
Stirling's formula states that for $x \rightarrow \infty$, we have $x! \approx x^x e^{-x} \sqrt{2\pi x}$. Substitute into the previous equation yields
\begin{align*}
\log p_X(x)
&\approx x \log \theta - \theta - \log \left(x^x e^{-x} \sqrt{2\pi x}\right)\\
&= x \log \theta - \theta - x\log x + x - \log \sqrt{2\pi x}.
\end{align*}
The Gaussian has to fit the Poisson well around the mean, which is $\theta$. Thus define $x = \theta +\epsilon$ with $\theta \gg \epsilon$. Then,
\begin{align*}
\log p_X(x)
&= {x} \log \theta - \theta - {x}\log {x} + {x} - \log \sqrt{2\pi {x}} \\
&\hspace{-7ex}= {(\theta+\epsilon)} \log \theta - \theta - {(\theta+\epsilon)}\log {(\theta+\epsilon)} \\
&\qquad + {(\theta+\epsilon)} - \log \sqrt{2\pi {(\theta+\epsilon)}} \\
&\hspace{-7ex}= \epsilon + (\theta+\epsilon)\log\frac{\theta}{\theta + \epsilon} - \log\sqrt{2\pi(\theta+\epsilon)}\\
&\hspace{-7ex}= \epsilon - (\theta+\epsilon)\log\left(1+\frac{\epsilon}{\theta}\right) - \log\sqrt{2\pi\theta}-\frac12\log\left(1+\frac{\epsilon}{\theta}\right)\\
&\hspace{-7ex}= \epsilon - \log\sqrt{2\pi\theta} - \left(\theta + \epsilon + \frac{1}{2} \right) {\log\left(1+\frac{\epsilon}{\theta}\right)}.
\end{align*}
For $\frac{\epsilon}{\theta}\ll 1$, it holds that $\log(1+\frac{\epsilon}{\theta}) \approx \frac{\epsilon}{\theta} - \frac{\epsilon^2}{2\theta^2} + \ldots$. Therefore,
\begin{align*}
\log p_X(x)
&\approx \epsilon - \log\sqrt{2\pi\theta} - \left(\theta + \epsilon + \frac{1}{2} \right) {\left(\frac{\epsilon}{\theta} - \frac{\epsilon^2}{2\theta^2} + \ldots\right)}\\
&= \epsilon - \log\sqrt{2\pi\theta} - \epsilon - \frac{\epsilon^2}{\theta} - \frac{\epsilon}{2\theta} + \frac{\epsilon^2}{2\theta} + \frac{\epsilon^2}{4\theta^2} + \ldots.
\end{align*}
By canceling terms, and removing $\frac{\epsilon^2}{4\theta^2}$ and $\frac{\epsilon}{2\theta}$ (because $\frac{\epsilon}{\theta} \ll 1$), it follows that $\log p_X(x) \approx -\frac{\epsilon^2}{2\theta} - \log\sqrt{2\pi\theta}$.
This implies that $p_X(x) \approx \frac{1}{\sqrt{2\pi \theta}}e^{-\frac{\epsilon^2}{2\theta}}$. Substituting $x = \theta + \epsilon$ completes the proof.
\end{proof}

\subsection{Mean Invariant = ML for Exponential Family}
A mean invariant estimator is generally not the same as an ML estimator. However, for distributions in the exponential family, they are identical. Recall the definition of the exponential family. A sequence of i.i.d. random variables $\mY = [Y_1,\ldots,Y_N]$ is said to be in the exponential family if
\begin{equation}
p_{\mY}(\vy;\theta) = h(\vy)\exp\left\{\eta(\theta) T(\vy) - A(\theta)\right\},
\end{equation}
for some functions $h$, $\eta$, $T$, and $A$ \cite{Lehmann_1999}. The function $T(\vy)$ is the sufficient statistic. Since a necessary condition for an estimator to be mean invariant is that it is a function of the sample average $\overline{Y}$, in the subsequent discussions it is assumed that $T(\vy) = \overline{y}$ where $\overline{y} = (1/N)\sum_{n=1}^N y_n$ with $y_1,\ldots,y_N$ as the realizations of $Y_1,\ldots,Y_N$.

\begin{example}
For one-bit QIS, the probability density function can be written as
\begin{align*}
p_{\mY}(\vy;\theta)
&= \prod_{n=1}^N \exp\left\{ \log\left[(1-e^{-\theta})^{y_n} (e^{-\theta})^{1-y_n}\right] \right\}\\
&= \exp\left\{ N \overline{y} \log(1-e^{-\theta}) + N(1-\overline{y})\log(e^{-\theta}) \right\}\\
&= \left[\exp\left\{ \overline{y} \log(e^\theta-1) -\theta \right\}\right]^N.
\end{align*}
So, one can associate $h(\vy) = 1$, $T(\vy) = \overline{y}$, $\eta(\theta) = \log(e^\theta-1)$, $A(\theta) = \theta$.
\end{example}

\begin{lemma}
\label{lemma: MI 1}
$\thetahat_{\text{ML}}(\vy)$ is the ML estimate if and only if it satisfies the equation
\begin{equation}
g(\thetahat_{\text{ML}}(\vy)) \bydef \frac{A'(\thetahat_{\text{ML}}(\vy))}{\eta'(\thetahat_{\text{ML}}(\vy))} = \overline{y},
\label{eq: lemma MI 1}
\end{equation}
where $g(\theta) \bydef A'(\theta)/\eta'(\theta)$ is a function, provided that $A'(\theta)$ and $\eta'(\theta)$ exist and are not zero. If $g^{-1}$ exists, then $\thetahat_{\text{ML}}(\vy) = g^{-1}(\overline{y})$.
\end{lemma}

\begin{proof}
$\thetahat_{\text{ML}}$ is the ML estimate if and only if $\frac{\partial }{ \partial \theta}\log p_{\mY}(\vy;\theta)\big|_{\theta = \thetahat_{\text{ML}}}  = 0$. The derivative is
\begin{align*}
\frac{\partial }{ \partial \theta} \log p_{\mY}(\vy;\theta)
&= \frac{\partial}{ \partial \theta} \left\{\log h(\vy) + \eta(\theta)\overline{y} - A(\theta)\right\} \\
&= \eta'(\theta)\overline{y} - A'(\theta).
\end{align*}
Setting it to zero, the condition is equivalent to $\frac{A'(\thetahat_{\text{ML}})}{\eta'(\thetahat_{\text{ML}})} = \overline{y}$. Defining $g(\theta) = \frac{A'(\theta)}{\eta'(\theta)}$, it follows that $g(\thetahat_{\text{ML}}(\vy)) = \overline{y}$. If $g^{-1}$ exists, then estimator is $\thetahat_{\text{ML}}(\mY) = g^{-1}(\overline{Y})$.
\end{proof}

\begin{lemma}
\label{lemma: MI 2}
For any distributions in the exponential family and for any parameter $\theta$,
\begin{equation}
\frac{A'(\theta)}{\eta'(\theta)} = \E[T(\mY)] \bydef \mu(\theta).
\label{eq: lemma MI 2}
\end{equation}
If $A'(\theta)$ and $\eta'(\theta)$ exist and are not zero, and $\mu^{-1}$ exists, then $\theta = \mu^{-1}(\E[T(\mY)])$.
\end{lemma}

\begin{proof}
The function $A(\theta)$ is the normalization constant, defined as $A(\theta) = \log \int h(\vy) e^{\eta(\theta) T(\vy)} d\vy$. Therefore, its derivative is
\begin{align*}
A'(\theta)
&= \frac{d}{d\theta} \left\{\log \int h(\vy) e^{\eta(\theta) T(\vy)} d\vy  \right\}\\
&= \frac{\int h(\vy) e^{\eta(\theta) T(\vy)} T(\vy) \eta'(\theta) d\vy}{\int h(\vy) e^{\eta(\theta) T(\vy)} \vy}\\
&= \eta'(\theta) \int T(\vy) \cdot p_{\mY}(\vy;\theta) d\vy = \eta'(\theta) \E[T(\mY)].
\end{align*}
Thus, if $\mu^{-1}$ exists, then $\theta = \mu^{-1}(\E[T(\mY)])$.
\end{proof}

Comparing \eref{eq: lemma MI 1} and \eref{eq: lemma MI 2}, it follows that the mapping is $g(\theta) = \mu(\theta)$. Therefore, $\mu(\thetahat_{\text{ML}}(\mY)) = g(\thetahat_{\text{ML}}(\mY))$. But since $g(\thetahat_{\text{ML}}(\mY)) = \overline{Y}$, it follows that $\mu(\thetahat_{\text{ML}}(\mY)) = \overline{Y}$. Thus, the ML estimator is mean invariant.

\subsection{Proof of Theorem 2}
\begin{proof}
Since $\overline{Y} = (1/N)\sum_{n=1}^N Y_n$, the mean $\E[\overline{Y}] = \E[Y_1]$. It then follows that the derivative $d\mu/d\theta$ remains unchanged. For the variance, it is easy to show that $\Var[\overline{Y}] = \Var[Y_1]/N$. Substituting these results into \eref{eq: SNR exp main} would yield
\begin{align*}
\text{SNR}_{\text{exp}}(\theta)
&= \frac{\theta}{\sqrt{\Var[\overline{Y}]}} \cdot \frac{d\mu}{d\theta} = \sqrt{N} \frac{\theta}{\sqrt{\Var[Y_1]}} \cdot \frac{d\mu}{d\theta}.
\end{align*}
Therefore, it suffices to prove the SNR for $N = 1$. In what remains, let $Y$ be any of the random variables $Y_1,\ldots,Y_N$ since they are i.i.d.

Recall the probability density function of $Y$:
\begin{align*}
p_Y(y) =
\begin{cases}
\frac{\theta^y}{y!} e^{-\theta}, &\qquad y < L,\\
\sum_{k=L}^{\infty} \frac{\theta^k}{k!} e^{-\theta} = 1-\Psi_L(\theta), &\qquad y \ge L,
\end{cases}
\end{align*}
where $\Psi_L(\theta)$ is the incomplete Gamma function. The mean of $Y$ can be shown as
\begin{align*}
\mu = \E[Y]
&= \sum_{k=0}^{L-1} k \cdot \frac{\theta^k}{k!} e^{-\theta} + L \cdot \left(\sum_{k=L}^{\infty} \frac{\theta^k}{k!} e^{-\theta}\right)\\
&= \sum_{k=1}^{L-1} \frac{\theta^k}{(k-1)!} e^{-\theta} + L \cdot (1-\Psi_{L}(\theta))\\
&= \theta \sum_{k=0}^{L-2} \frac{\theta^{k-1}}{k!} e^{-\theta} + L \cdot (1-\Psi_{L}(\theta))\\
&= \theta \Psi_{L-1}(\theta) + L \cdot (1-\Psi_{L}(\theta)).
\end{align*}
The derivative $d\mu/d\theta$ is therefore
\begin{align*}
\frac{d\mu}{d\theta}
&= \frac{d}{d\theta}\left\{\theta \Psi_{L-1}(\theta) + L \cdot (1-\Psi_{L}(\theta))\right\} \\
&= \theta \Psi_{L-1}'(\theta) + \Psi_{L-1}(\theta) - L \cdot \Psi'_L(\theta).
\end{align*}

For the variance, since $\Var[Y] = \E[Y^2] - \mu^2$, it remains to determine $\E[Y^2]$.
\begin{align*}
\E[Y^2]
&= \sum_{k=0}^{L-1} k^2 \cdot \frac{\theta^k}{k!} e^{-\theta} + L^2 \cdot \left(\sum_{k=L}^{\infty} \frac{\theta^k}{k!} e^{-\theta}\right)\\
&= \sum_{k=1}^{L-1} k \frac{\theta^k}{(k-1)!} e^{-\theta} + L^2 \cdot (1-\Psi_{L}(\theta))\\
&= \sum_{k=1}^{L-1} (k - 1 + 1) \frac{\theta^k}{(k-1)!} e^{-\theta} + L^2 \cdot (1-\Psi_{L}(\theta))\\
&= \sum_{k=2}^{L-1} \frac{\theta^{k}}{(k-2)!} e^{-\theta} + \sum_{k=1}^{L-1} \frac{\theta^{k}}{(k-1)!} e^{-\theta} \\
&\qquad\qquad\qquad\qquad\qquad + L^2 \cdot (1-\Psi_{L}(\theta))\\
&= \theta^2 \Psi_{L-2}(\theta) + \theta\Psi_{L-1}(\theta) + L^2(1-\Psi_L(\theta)).
\end{align*}
This completes the proof.
\end{proof}

\subsection{Proof of Corollary~\ref{corollary: limiting case}}
\begin{proof}
When $L$ is large, $\Psi_L(10^\phi)$ and $\Psi_{L-1}(10^\phi)$ are close enough that they can be considered approximately equal. Denote the value $\Psi_L(10^\phi)$ as $\Psi$. Then by \eref{eq: limit Psi} it holds that $\Psi \rightarrow 0$ for $\phi > \log_{10}L$ and $\Psi \rightarrow 1$ for $\phi \le \log_{10}L$. In either case, since $\Psi$ is a constant, it follows that the derivative $\Psi'_L(10^\phi) = 0$ as long as $\phi > \log_{10}L$ or $\phi < \log_{10}L$. Therefore, by denoting $\theta = 10^\phi$, the two cases can be derived as follows.

When $\phi \le \log_{10}L$, it holds that
\begin{align*}
\mu     &= \theta \Psi + L (1-\Psi) = \theta,\\
\Var[Y] &= \theta^2 \Psi + \theta\Psi + L^2(1-\Psi) - \mu^2 \\
        &= \theta^2 \cdot 1 + \theta \cdot 1 + L^2 \cdot 0 - \mu^2 = \theta,\\
\frac{d\mu}{ d\theta }  &= \Psi - (L-\theta)\Psi_{L-1}'(\theta)\\
                        &= 1 - (L-\theta)\cdot 0 = 1.
\end{align*}
So, the SNR for $\phi \le \log_{10}L$ is
\begin{align*}
20\log_{10}\text{SNR}_{\exp}(10^\phi) = 20\log_{10}\sqrt{10^\phi} = 10 \phi.
\end{align*}

When $\phi > \log_{10}L$, $\Psi \rightarrow 0$. Therefore,
\begin{align*}
\mu     &= \theta \Psi + L (1-\Psi) \approx \theta \Psi + L,\\
\Var[Y] &= \theta^2 \Psi + \theta\Psi + L^2(1-\Psi) - \mu^2 \\
        &= \theta^2 \Psi + \theta\Psi + L^2(1-\Psi) - (\theta \Psi + L)^2\\
        &= \theta^2 \Psi + \theta\Psi + L^2 - \theta^2 \Psi^2 - 2\theta\Psi L -  L^2 \\
        &= \theta^2 \Psi + \theta\Psi - 2\theta\Psi L\\
        &= \theta\Psi(\theta + 1 - 2L),\\
\frac{d\mu}{ d\theta }  &= \theta \Psi_{L-1}'(\theta) + \Psi_{L-1}(\theta)- L\Psi_{L}'(\theta)\\
                        &= \theta \cdot 0 + \Psi- L \cdot 0 = \Psi.
\end{align*}
By taking the limit that $\Psi \rightarrow 0$, it follow that
\begin{align*}
\lim_{\Psi\rightarrow 0} \text{SNR}_{\exp}(10^\phi)
&= \lim_{\Psi\rightarrow 0} \frac{\theta}{\sqrt{\Var[Y]}} \cdot \frac{d\mu}{ d\theta } \\
&= \lim_{\Psi\rightarrow 0} \frac{\theta}{\sqrt{\theta\Psi(\theta + 1 - 2L)}} \cdot \Psi \\
&= \lim_{\Psi\rightarrow 0} \frac{\sqrt{\theta \Psi}}{\sqrt{(\theta + 1 - 2L)}} = 0.
\end{align*}
Combining with the case where $\phi \le \log_{10}L$, the overall SNR is proved.
\end{proof}

\subsection{MATLAB Code for Monte Carlo Simulation}
The MATLAB code below illustrates the Monte Carlo simulation of how $\text{SNR}_{\text{exp}}(\theta)$ is generated for a truncated Poisson distribution. Adding other factors to the forward model can be done by modifying the random variable $Y$.
\begin{Verbatim}[frame=single,framerule=0.2mm]
N  = 100000;
L  = 10;
theta_set = logspace(-2,3,100);
mu    = zeros(1,100);
sigma = zeros(1,100);
for i=1:100
    theta    = theta_set(i);
    Theta    = theta*ones(N,1);
    Y        = poissrnd(Theta);
    Y(Y>L)   = L;
    mu(i)    = mean(Y);
    sigma(i) = std(Y);
end
dmu_dt  = [diff(mu)./diff(theta_set) 1];
SNR     = theta_set./sigma.*dmu_dt;
loglog(theta_set, SNR);
\end{Verbatim}

For plotting the theoretical $\text{SNR}_{\text{exp}}(\theta)$, one just needs to call the incomplete Gamma function.
\begin{Verbatim}[frame=single,framerule=0.2mm]
theta = logspace(-2,3,100);
Psi   = gammainc(theta,L,'upper');
Psi1  = gammainc(theta,L-1,'upper');
Psi2  = gammainc(theta,L-2,'upper');
dPsi  = -theta.^(L-1).*exp(-theta) ...
        /gamma(L);
dPsi1 = -theta.^(L-2).*exp(-theta) ...
        /gamma(L-1);
mu    = theta.*Psi1 + L.*(1 - Psi);
sigma = sqrt(theta.^2 .* Psi2 + ...
        theta.*Psi1 + L^2*(1-Psi) - ...
        mu_theory.^2);
dmu_dt = theta.*dPsi1 + Psi1 - L*dPsi;
SNR    = theta./sigma.*dmu_dt;
loglog(theta, SNR);
\end{Verbatim}
The combination of the two pieces of codes, with minor modifications, is sufficient to reproduce all the figures reported in this paper.

\bibliography{ref}

\begin{thebibliography}{10}

\bibitem{Shannon_1949}
C.~Shannon, ``Communication in the presence of noise,'' {\em Proceedings of the
  IRE}, vol.~37, no.~1, pp.~10--21, 1949.

\bibitem{EMVA_2010}
E.~M.~V. Association, ``{EMVA Standard 1288}: Standard for characterization of
  image sensors and cameras,'' 2010.
\newblock Available online at
  \url{https://www.emva.org/wp-content/uploads/EMVA1288-3.0.pdf}.

\bibitem{Gamal_Lecture}
A.~E. Gamal, ``{EE392B Lecture Note: SNR and Dynamic Range}.''
\newblock \url{https://isl.stanford.edu/~abbas/ee392b/lect08.pdf}, accessed
  12/7/2021.

\bibitem{Mitsunaga_Nayar}
T.~Mitsunaga and S.~K. Nayar, ``Radiometric self calibration,'' in {\em IEEE
  CVPR}, vol.~1, pp.~374--380, 1999.

\bibitem{elgendy2018optimal}
O.~A. Elgendy and S.~H. Chan, ``Optimal threshold design for {Q}uanta {I}mage
  {S}ensor,'' {\em IEEE Trans. Computational Imaging}, vol.~4, no.~1,
  pp.~99--111, 2018.

\bibitem{fossum2013modeling}
E.~R. Fossum, ``Modeling the performance of single-bit and multi-bit quanta
  image sensors,'' {\em IEEE J. Electron Devices Society}, vol.~1, no.~9,
  pp.~166--174, 2013.

\bibitem{yang2011bits}
F.~Yang, Y.~M. Lu, L.~Sbaiz, and M.~Vetterli, ``Bits from photons: Oversampled
  image acquisition using binary {P}oisson statistics,'' {\em IEEE Trans. Image
  Process.}, vol.~21, no.~4, pp.~1421--1436, 2012.

\bibitem{Nakamura_2005_book}
J.~Nakamura, {\em Image Sensors and Signal Processing for Digital Still
  Cameras}.
\newblock CRC Press, Talyor and Francis Group, 2005.

\bibitem{lim2006characterization}
S.-H. Lim, ``Characterization of noise in digital photographs for image
  processing,'' in {\em Digital Photography II}, vol.~6069, p.~60690O,
  International Society for Optics and Photonics, 2006.

\bibitem{fossum200611}
E.~R. Fossum, ``Some thoughts on future digital still cameras,'' in {\em Image
  Sensors and Signal Processing for Digital Still Cameras}, p.~305, CRC, 2006.

\bibitem{Whittlesey}
J.~R.~B. Whittlesey, ``Incomplete gamma functions for evaluating erlang process
  probabilities,'' {\em Mathematics of Computation}, vol.~17, no.~81,
  pp.~11--17, 1963.

\bibitem{Lehmann_1999}
E.~L. Lehmann, {\em Elements of Large-Sample Theory}.
\newblock Springer, 1999.

\bibitem{Gnanasambandam_TCI_HDR}
A.~Gnanasambandam and S.~H. Chan, ``{HDR} imaging with {Quanta Image Sensors}:
  Theoretical limits and optimal reconstruction,'' {\em IEEE Trans.
  Computational Imaging}, vol.~6, pp.~1571--1585, 2020.

\bibitem{Gupta_2019}
A.~Ingle, A.~Velten, and M.~Gupta, ``High flux passive imaging with
  single-photon sensors,'' in {\em IEEE CVPR}, pp.~6760--6769, 2019.

\bibitem{McCullagh_Nelder_1983}
P.~McCullagh and J.~A. Nelder, {\em Generalized linear models}.
\newblock Chapman and Hall, 1983.

\bibitem{granados2010optimal}
M.~Granados, B.~Ajdin, M.~Wand, C.~Theobalt, H.-P. Seidel, and H.~P. Lensch,
  ``Optimal {HDR} reconstruction with linear digital cameras,'' in {\em IEEE
  CVPR}, pp.~215--222, 2010.

\bibitem{Gamal_2005_magazine}
A.~El~Gamal and H.~Eltoukhy, ``{CMOS} image sensors,'' {\em IEEE Circuits and
  Devices Magazine}, vol.~21, no.~3, pp.~6--20, 2005.

\bibitem{Ma:17}
J.~Ma, S.~Masoodian, D.~A. Starkey, and E.~R. Fossum, ``Photon-number-resolving
  megapixel image sensor at room temperature without avalanche gain,'' {\em OSA
  Optica}, vol.~4, pp.~1474--1481, Dec 2017.

\bibitem{mann94b}
S.~Mann and R.~Picard, ``Being `undigital' with digital cameras: Extending
  dynamic range by combining differently exposed pictures,'' Tech. Rep. 323,
  M.I.T. Media Lab Perceptual Computing Section, Boston, Massachusetts, 1994.

\bibitem{Debevec_Malik_HDR}
P.~Debevec and J.~Malik, ``Recovering high dynamic range radiance maps from
  photographs,'' in {\em ACM SIGGRAPH}, pp.~369--378, 1997.

\bibitem{nayar2000high}
S.~K. Nayar and T.~Mitsunaga, ``High dynamic range imaging: Spatially varying
  pixel exposures,'' in {\em IEEE CVPR}, vol.~1, pp.~472--479, 2000.

\bibitem{kirk2006noise}
K.~Kirk and H.~J. Andersen, ``Noise characterization of weighting schemes for
  combination of multiple exposures,'' in {\em British Machine Vision Conf.},
  pp.~1129--1138, 2006.

\bibitem{Gnanasambandam_Megapixel_2019}
A.~Gnanasambandam, O.~Elgendy, J.~Ma, and S.~H. Chan, ``Megapixel
  photon-counting color imaging using {Q}uanta {I}mage {S}ensor,'' {\em OSA
  Optics Express}, vol.~27, no.~12, pp.~17298--17310, 2019.

\bibitem{Elgendy_CFA_2021}
O.~A. Elgendy and S.~H. Chan, ``Color {F}ilter {A}rrays for {Q}uanta {I}mage
  {S}ensors,'' {\em IEEE Trans. Computational Imaging}, vol.~6, pp.~652--665,
  Jan 2020.

\bibitem{Elgendy_demosaicking_2021}
O.~A. Elgendy, A.~Gnanasambandam, S.~H. Chan, and J.~Ma, ``Low-light
  demosaicking and denoising for small pixels using learned frequency
  selection,'' {\em IEEE Transa. Computational Imaging}, vol.~7, pp.~137--150,
  2021.

\bibitem{Chan_Density_2022}
S.~H. Chan, ``On the insensitivity of bit density to read noise in one-bit
  {Quanta Image Sensors},'' 2022.
\newblock Available online at \url{http://arxiv.org/abs/2203.06086}, accessed
  6/12/2022.

\bibitem{Chan_GlobalSIP_2014}
S.~H. Chan and Y.~M. Lu, ``Efficient image reconstruction for gigapixel
  {Q}uantum {I}mage {S}ensors,'' in {\em IEEE Global Conf. Signal and Info.
  Process.}, pp.~312--316, 2014.

\bibitem{Chan_PnP_2016}
S.~H. Chan, X.~Wang, and O.~A. Elgendy, ``Plug-and-play {ADMM} for image
  restoration: Fixed-point convergence and applications,'' {\em IEEE Trans.
  Computational Imaging}, vol.~3, pp.~84--98, Nov. 2016.

\bibitem{Chan_MDPI_2016}
S.~H. Chan, O.~A. Elgendy, and X.~Wang, ``Images from bits: {Non-iterative}
  image reconstruction for {Q}uanta {I}mage {S}ensors,'' {\em MDPI Sensors},
  vol.~16, no.~11, p.~1961, 2016.

\bibitem{Choi_ICASSP_2018}
J.~H. Choi, O.~A. Elgendy, and S.~H. Chan, ``Image reconstruction for {Q}uanta
  {I}mage {S}ensors using deep neural networks,'' in {\em IEEE ICASSP},
  pp.~6543--6547, 2018.

\bibitem{chi_chan_ECCV2020}
Y.~Chi, A.~Gnanasambandam, V.~Koltun, and S.~H. Chan, ``Dynamic low-light
  imaging with {Quanta Image Sensors},'' in {\em ECCV}, pp.~122--138, 2020.

\bibitem{gnanasambandam_chan_ECCV2020}
A.~Gnanasambandam and S.~H. Chan, ``Image classification in the dark using
  {Quanta Image Sensors},'' in {\em ECCV}, pp.~484--501, 2020.

\bibitem{Chengxi}
C.~Li, X.~Qu, A.~Gnanasambandam, O.~A. Elgendy, J.~Ma, and S.~H. Chan,
  ``Photon-limited object detection using non-local feature matching and
  knowledge distillation,'' in {\em ICCV-W}, pp.~3959--3970, 2021.

\end{thebibliography}
\bibliographystyle{ieeetr}

\end{document}